\documentclass[11pt, a4paper]{article}

\usepackage[whole]{bxcjkjatype}

\usepackage{lipsum}
\usepackage{amsfonts}
\usepackage{mathtools}
\usepackage{amsmath}
\usepackage{amssymb}
\usepackage{amsfonts}
\usepackage{multirow}
\usepackage{bm} 
\usepackage{epstopdf}
\usepackage{float}
\usepackage{natbib}
\usepackage{amsmath}
\usepackage{amsfonts}
\usepackage{amsthm}
\usepackage{url}

\mathtoolsset{showonlyrefs=true}
\DeclarePairedDelimiter\paren{\lparen}{\rparen}

\usepackage{algorithm,algorithmic} 

\usepackage{tikz}
\usepackage{pgfplotstable}
\usetikzlibrary{arrows,positioning,plotmarks,external,patterns,angles,
decorations.pathmorphing,backgrounds,fit,shapes,graphs,calc,spy}
\pgfplotsset{compat=1.14}

\DeclareMathOperator*{\argmin}{argmin}
\DeclareMathOperator*{\argmax}{argmax}

\newcommand{\rmd}{\mathrm d}

\newcommand{\rmN}{\mathrm N}
\newcommand{\bbR}{\mathbb R}

\newtheorem{theorem}{Theorem}
\newtheorem{lemma}{Lemma}
\newtheorem{proposition}{Proposition}
\newtheorem{corollary}{Corollary}
\newtheorem{remark}{Remark}
\newtheorem{algoalgo}{Algorithm}

\usepackage{typearea}
\typearea{12}

\begin{document}

\title{Piecewise monotone estimation\\ in one-parameter exponential family}
	
\author{Takeru Matsuda\thanks{Department of Mathematical Informatics, University of Tokyo \& Statistical Mathematics Unit, RIKEN Center for Brain Science, e-mail: 
		\texttt{matsuda@mist.i.u-tokyo.ac.jp}} \and Yuto Miyatake\thanks{Cybermedia Center, Osaka University, e-mail: 
		\texttt{miyatake@cas.cmc.osaka-uac.jp}}}
		
\date{}

\maketitle

\begin{abstract}
The problem of estimating a piecewise monotone sequence of normal means is called the nearly isotonic regression.
For this problem, an efficient algorithm has been devised by modifying the pool adjacent violators algorithm (PAVA).
In this study, we investigate estimation of a piecewise monotone parameter sequence for general one-parameter exponential families such as binomial, Poisson and chi-square.
We develop an efficient algorithm based on the modified PAVA, which utilizes the duality between the natural and expectation parameters.
We also provide a method for selecting the regularization parameter by using an information criterion.
Simulation results demonstrate that the proposed method detects change-points in piecewise monotone parameter sequences in a data-driven manner.
Applications to spectrum estimation, causal inference and discretization error quantification of ODE solvers are also presented.
\end{abstract}

	\section{Introduction}
There are many phenomena that involve monotonicity, such as the dose-response curve in medicine and the demand/supply curves in economics.
Parameter estimation under such order constraints is a typical example of shape constrained inference \citep{bb72,rwd88,vE06,groeneboom2014nonparametric}.
For example, suppose that we have $n$ normal observations $X_i \sim {\rm N}(\mu_i,1)$ for $i=1,\dots,n$, where $\mu_1 \leq \mu_2 \leq \cdots \leq \mu_n$ is a monotone sequence of normal means.
In this setting, the maximum likelihood estimate (MLE) of $\mu$ is the solution of the constrained optimization
\begin{align}
	\hat{\mu} = \argmin_{\mu_1 \leq \cdots \leq \mu_n} \frac{1}{2} \sum_{i=1}^n (X_i-\mu_i)^2, \label{iso}
\end{align}
which coincides with the isotonic regression of $X_1,\dots,X_n$ with uniform weights and  efficiently solved by the pool adjacent violators algorithm (PAVA) ~\citep[Chapter~1]{rwd88}.
Statistical properties of isotonic regression estimators have been extensively studied such as the convergence rates and risk bounds \citep{bellec2018sharp,groeneboom2014nonparametric,guntuboyina2018nonparametric,han2019isotonic}.

Whereas isotonic regression is useful for estimating a monotone sequence of normal means, the order constraint may be violated at a few change-points in practice.
In other words, the parameter sequence may be only piecewise monotone.
Thus, \cite{tibshirani2011nearly} investigated the problem of estimating a piecewise monotone sequence of normal means and called it the nearly isotonic regression.
Specifically, for $n$ (homoscedastic) normal observations $X_i \sim {\rm N}(\mu_i,1)$ for $i=1,\dots,n$, they formulated the problem as the regularized optimization given by
\begin{align}
	\hat{\mu}_{\lambda} = \argmin_{\mu} \frac{1}{2} \sum_{i=1}^n (X_i - \mu_i)^2 + \lambda \sum_{i=1}^{n-1} (\mu_i-\mu_{i+1})_+, \label{org_neariso}
\end{align}
where $(a)_+=\max(a,0)$ and $\lambda>0$ is the regularization parameter.
Then, they developed an efficient algorithm for this problem by modifying the PAVA.
They also showed that the number of joined pieces provides an unbiased estimate of the degrees of freedom, which enables data-driven selection of the regularization parameter $\lambda$.

In this study, we investigate estimation of a piecewise monotone parameter sequence for general one-parameter exponential families \citep{Efron} including (heteroscedastic) normal, binomial, Poisson and (scaled) chi-square. 
Suppose that we have $n$ observations $X_i \sim p_i(x_i \mid \theta_i)$ for $i=1,\dots,n$, where each $p_i(x_i \mid \theta_i)$ is a one-parameter exponential family defined by
\begin{align*}
	p_i(x_i \mid \theta_i) = h_i(x_i) \exp (\theta_i x_i - w_i \psi(\theta_i)).
\end{align*}
For example, the binomial distribution ${\rm Bi}(N_i,r_i)$ with $N_i$ (fixed) trials of success probability $r_i$ corresponds to $x_i \in \{ 0,1,\dots,N_i \}$, $h_i(x_i) = N_i!/(x_i! (N_i-x_i)!)$, $w_i=N_i$ and $\psi(\theta_i) = \log (1+e^{\theta_i})$, where $r_i=e^{\theta_i}/(1+e^{\theta_i})$.
The Poisson distribution ${\rm Po}(\lambda_i)$ with mean $\lambda_i$ corresponds to $x_i \in \{ 0,1,\dots \}$, $h_i(x_i) = 1/(x_i!)$, $w_i=1$ and $\psi(\theta_i) = e^{\theta_i}$, where $\lambda_i=e^{\theta_i}$.
The scaled chi-square distribution $s_i \chi^2(d_i)$ with scale $s_i$ and $d_i$ degrees of freedom corresponds to $x_i \geq 0$, $h_i(x_i) = x_i^{d_i/2-1}/\Gamma(d_i/2)$, $w_i=d_i/2$ and $\psi(\theta_i) = -\log (-\theta_i)$, where $s_i=-1/(2 \theta_i)$.
To estimate the piecewise monotone sequence $\theta=(\theta_1,\dots,\theta_n)$, we consider the regularized estimator  defined by
\begin{align}
	\hat{\theta}_{\lambda} &= \argmin_{\theta} -\sum_{i=1}^n \log p_i(X_i \mid \theta_i) + \lambda \sum_{i=1}^{n-1} (\theta_i-\theta_{i+1})_+ \nonumber \\
	&= \argmin_{\theta} \sum_{i=1}^n (- \theta_i X_i + w_i \psi(\theta_i)) + \lambda \sum_{i=1}^{n-1} (\theta_i-\theta_{i+1})_+, \label{ggneariso}
\end{align}
where $\lambda > 0$ is the regularization parameter.
We develop an efficient algorithm for this optimization problem by extending the modified PAVA and utilizing the duality between the natural parameters $\theta_i$ and expectation parameters $\eta_i={\rm E}_{\theta_i} [X_i]=w_i \psi'(\theta_i)$.
We also provide a method for selecting the regularization parameter $\lambda$ by using an information criterion.
Simulation results demonstrate that the proposed method successfully detects change-points of $\theta$ in a data-driven manner.
We present applications to spectrum estimation, causal inference and discretization error quantification of ODE solvers.

This paper is organized as follows.
In Section~\ref{sec:main}, we develop a method for piecewise monotone estimation in general one-parameter exponential familes.
In Section~\ref{sec:simulation}, simulation results are presented.
In Section~\ref{sec:application}, applications to spectrum estimation, causal inference and discretization error quantification are presented.
In Section~\ref{sec:conclusion}, concluding remarks are given.
In Appendix, a brief review on (nearly) isotonic regression, technical proofs, and additional experiments are provided.
A julia package of the proposed method is available online at \url{https://github.com/yutomiyatake/IsoFuns.jl}.

\section{Proposed method}\label{sec:main}
\subsection{Estimation algorithm}
We propose the following algorithm for computing the regularization path of the estimator \eqref{ggneariso}.
This algorithm outputs the set of critical points (knots) $\lambda_0=0,\lambda_1,\dots,\lambda_T$ and the estimate $\hat{\eta}_{\lambda_t}=\psi'(\hat{\theta}_{\lambda_t})$ at each critical point.
Note that $\psi'$ is monotonically increasing because $\psi$ is strictly convex for exponential families \citep{Efron}.
Since the solution $\hat{\eta}_{\lambda}=\psi'(\hat{\theta}_{\lambda})$ is piecewise linear with respect to $\lambda$ as shown below (Theorem~\ref{th:weighted}), the solution for general $\lambda$ is readily obtained by linear interpolation.

\begin{algoalgo}\label{alg:wmpava}
	\mbox{}
	\begin{itemize}
		\item Input: $z \in \mathbb{R}^n$ (observation), $w \in \mathbb{R}^n$ (weight)
		
		\item Output: $\lambda_1,\dots,\lambda_T$ (knot), $\hat{\eta}_{\lambda_1},\dots,\hat{\eta}_{\lambda_T}  \in \mathbb{R}^n$ (estimate)
		
		\item Start with $t=0$, $\lambda_0=0$, $\hat{\eta}_{\lambda_0}=z$ and $K=n$ clusters $A_j=\{ j \}$ with value $y_{A_j}=z_j$ for $j=1,\dots,n$.
		
		\item Repeat:
		\begin{itemize}
			\item Set {$s_0=s_K=0$} and {$s_j=I(\hat{\eta}_{\lambda_{t},\min A_j}-\hat{\eta}_{\lambda_{t},\max A_j+1}>0)$} for {$j=1,\dots,K-1$}.
			
			\item Compute $m_j=(s_{j-1}-s_j)/(\sum_{i \in A_j} w_i)$ for $j=1,\dots,K$.
			
			\item Compute {$t_{j,j+1}=\lambda_t+(y_{A_{j+1}}-y_{A_j})/(m_j-m_{j+1})$} for $j=1,\dots,K-1$.
			
			\item If {$t_{j,j+1} \leq \lambda_t$} for every $j$, then terminate.
			
			\item Set {$j_*=\argmin_j \{ t_{j,j+1} \mid t_{j,j+1}>\lambda_t \}$} and {$\lambda_{t+1}=t_{j_*,j_*+1}$}.
			
			\item Update $y_{A_j}$ to $y_{A_j}+m_j(\lambda_{t+1}-\lambda_t)$ and set $\hat{\eta}_{\lambda_{t+1},i}=y_{A_j}$ for $i \in A_j$.
			
			\item Merge $A_{j_*+1}$ into $A_{j_*}$ and renumber $A_{j_*+2},\dots,A_K$ to $A_{j_*+1},\dots,A_{K-1}$.
			
			\item Decrease $K$ by one and increase $t$ by one.
		\end{itemize}
	\end{itemize}
\end{algoalgo}

Algorithm~\ref{alg:wmpava} can be viewed as a weighted version of the modified PAVA by \cite{tibshirani2011nearly}, where the total weight $\sum_{i \in A_j} w_i$ replaces the cardinality $|A_j|$ for each $A_j$.

\begin{remark}
	Algorithm~\ref{alg:wmpava} can be intuitively understood by using a physical model of inelastic collisions \citep[cf. ][]{sibuya1990equipartition}.
	Suppose that there are $n$ free particles of mass 1 moving on the one-dimensional line, which are numbered $1,\dots,n$ from the left, and the $i$th particle has mass $w_i$ and velocity $z_i$, whose sign represents the direction of the motion.
	The particles form clusters by perfectly inelastic collisions.
	For example, if the first and second particles collide ($z_1>z_2$), then they stick together and become a cluster of mass $w_1+w_2$ and velocity $(w_1 z_1+w_2 z_2)/(w_1+w_2)$.
	In general, if a cluster of mass $m_1$ and velocity $y_1$ collides with another cluster of mass $m_2$ and velocity $y_2<y_1$, then they form a cluster of mass $m_1+m_2$ and velocity $(m_1 y_1+m_2 y_2)/(m_1+m_2)$.
	Then, collisions cease within a finite time and eventually the particles are grouped into several clusters.
	Algorithm~\ref{alg:wmpava} can be viewed as simulating these collusions, where the regularization parameter $\lambda$ specifies the elapsed time from the beginning and each critical point $\lambda_t$ corresponds to the moment of collision.
	
	For simplicity, we assumed that a simultaneous collision of more than two clusters does not occur in Algorithm~\ref{alg:wmpava}.
	While this assumption is satisfied almost surely for continuous distributions such as Gaussian and chi-square, it may be violated for discrete distributions such as binomial and Poisson.
	Our julia package at \url{https://github.com/yutomiyatake/IsoFuns.jl} deals with such collisions properly.
\end{remark}

The following lemma is critical in showing the validity of Algorithm~\ref{alg:wmpava}.
Its proof is given in Appendix.

\begin{lemma}\label{lem:merge}
	If $(\hat{\theta}_{\bar{\lambda}})_i=(\hat{\theta}_{\bar{\lambda}})_{i+1}$ for some $\bar{\lambda}$, then $(\hat{\theta}_{\lambda})_i=(\hat{\theta}_{\lambda})_{i+1}$ for every $\lambda \geq \bar{\lambda}$.
\end{lemma}

From Lemma~\ref{lem:merge}, we obtain the following theorem by using a similar argument to \cite{friedman2007pathwise} and \cite{tibshirani2011nearly}.

\begin{theorem}\label{th:weighted}
	Let $\lambda_1,\dots,\lambda_T$ and $\hat{\eta}_{\lambda_1},\dots,\hat{\eta}_{\lambda_T}$ be the output of Algorithm~\ref{alg:wmpava} on $w=(w_1,\dots,w_n)$ and $z=(X_1/w_1,\dots,X_n/w_n)$. 
	Then, the estimator $\hat{\theta}_{\lambda}$ with $\lambda \in [\lambda_t,\lambda_{t+1}]$ in \eqref{ggneariso} is given by 
	\begin{align}
		(\hat{\theta}_{\lambda})_i = (\psi')^{-1} \left( \frac{\lambda_{t+1}-\lambda}{\lambda_{t+1}-\lambda_t} (\hat{\eta}_{\lambda_t})_i + \frac{\lambda-\lambda_t}{\lambda_{t+1}-\lambda_t} (\hat{\eta}_{\lambda_{t+1}})_i \right) \label{interpolate}
	\end{align}
	for $i=1,\dots,n$.
\end{theorem}
\begin{proof}
		
		Since the objective function of the optimization in \eqref{ggneariso} is strictly convex, it has a unique solution satisfying the subgradient condition \citep{bertsekas1997nonlinear}:
		\begin{align}
			-x_i + w_i \psi'((\hat{\theta}_{\lambda})_i) + \lambda (\xi_i-\xi_{i-1}) = 0 \label{KKT1}
		\end{align}
		for $i=1,\dots,n$, where $\xi_0=\xi_n=0$ and
		\begin{align}
			\xi_i &
			\begin{cases}
				= 1 & ((\hat{\theta}_{\lambda})_i > (\hat{\theta}_{\lambda})_{i+1}) \\
				= 0 & ((\hat{\theta}_{\lambda})_i < (\hat{\theta}_{\lambda})_{i+1})\\
				\in [0,1] & ((\hat{\theta}_{\lambda})_i = (\hat{\theta}_{\lambda})_{i+1})
			\end{cases} \label{KKT2}
		\end{align}
		for $i=1,\dots,n-1$.
		We show that \eqref{interpolate} satisfies \eqref{KKT1} in the following.
		
		At $\lambda=0$, the solution of \eqref{KKT1} is clearly given by $(\hat{\theta}_0)_i=(\psi')^{-1}(x_i/w_i)$ for $i=1,\dots,n$.
		From the initial condition of Algorithm~\ref{alg:wmpava}, it coincides with \eqref{interpolate} with $\lambda=0$ and $t=0$.
		
		Suppose that $\hat{\theta}_{\lambda}$ is clustered into $A_1,\dots,A_K$ at $\bar{\lambda} \geq 0$: $(\hat{\theta}_{\bar{\lambda}})_i = (\hat{\theta}_{\bar{\lambda}})_{A_j}$ for $i \in A_j$ and $j=1,\dots,K$, where $(\hat{\theta}_{\bar{\lambda}})_{A_j} \neq (\hat{\theta}_{\bar{\lambda}})_{A_{j+1}}$ for $j=1,\dots,K-1$.
		We consider the change of $\hat{\theta}_{\lambda}$ as $\lambda$ increases from $\bar{\lambda}$.
		From Lemma~\ref{lem:merge}, the clustering structure of $\hat{\theta}_{\lambda}$ remains the same and thus $\xi_1,\dots,\xi_n$ are constant until some neighboring clusters merge.
		Thus, by summing up \eqref{KKT1} for $i \in A_j$, we find that $\hat{\theta}_{\lambda}$ changes linearly with respect to $\lambda$ as long as the clustering structure remains the same:
		\begin{align}
			\psi'((\hat{\theta}_{\lambda})_{A_j}) = \left( \sum_{i \in A_j} w_i \right)^{-1} \left( \sum_{i \in A_j} x_i - \lambda (\xi_{\max A_j}-\xi_{\min A_j - 1}) \right), \label{linear}
		\end{align}
		which yields
		\begin{align}
			\psi'((\hat{\theta}_{\lambda})_{A_j}) - \psi'((\hat{\theta}_{\bar{\lambda}})_{A_j}) = \left( \sum_{i \in A_j} w_i \right)^{-1}  (\lambda-\bar{\lambda}) (\xi_{\min A_j - 1}-\xi_{\max A_j}).
		\end{align}
		Therefore, two clusters $A_j$ and $A_{j+1}$ merge ($(\hat{\theta}_{\lambda})_{A_j}=(\hat{\theta}_{\lambda})_{A_{j+1}}$) at
		\begin{align}
			\lambda = \bar{\lambda} + \frac{\psi'((\hat{\theta}_{\bar{\lambda}})_{A_{j+1}}) - \psi'((\hat{\theta}_{\bar{\lambda}})_{A_j})}{m_j-m_{j+1}}, \label{merge_lambda}
		\end{align}
		with the merged value
		\begin{align}
			(\hat{\theta}_{\lambda})_{A_j} = (\hat{\theta}_{\lambda})_{A_{j+1}} = (\psi')^{-1} \left( \psi'((\hat{\theta}_{\bar{\lambda}})_{A_j}) + m_j (\lambda-\bar{\lambda}) \right), \label{merge_value}
		\end{align}
		where
		\begin{align}
			m_j = \frac{\xi_{\min A_j - 1}-\xi_{\max A_j}}{\sum_{i \in A_j} w_i}.
		\end{align}
		By putting $y_{A_j} = \psi'((\hat{\theta}_{\bar{\lambda}})_{A_j})$ and $s_j = \xi_{\max A_j}$, the second term in \eqref{merge_lambda} coincides with $(y_{A_{j+1}}-y_{A_j})/(m_j-m_{j+1})$ in Algorithm~\ref{alg:wmpava} and \eqref{merge_value} coincides with the updated value of $y_{A_j}$ in Algorithm~\ref{alg:wmpava}.
		Hence, Algorithm~\ref{alg:wmpava} computes the change-points $\lambda_1,\dots,\lambda_T$ of the clustering structure and the solution of \eqref{KKT1} at each $\lambda_1,\dots,\lambda_T$ correctly.
		From the linear interpolation property \eqref{linear}, the solution of \eqref{KKT1} for general $\lambda \in [\lambda_t,\lambda_{t+1}]$ is given by \eqref{interpolate}.
	\end{proof}
	
	
	Here, we summarize the specializations of Theorem~\ref{th:weighted} to the heteroscedastic normal, binomial, Poisson and chi-square for convenience.
	
	\begin{corollary}
		Let $X_i \sim {\rm N}(\mu_i,\sigma_i^2)$ for $i=1,\dots,n$.
		Then, the estimator
		\begin{align}
			\hat{\mu}_{\lambda} &= \argmin_{\mu} -\sum_{i=1}^n \log p_i(X_i \mid \mu_i) + \lambda \sum_{i=1}^{n-1} (\mu_i-\mu_{i+1})_+ \label{hetero}
		\end{align}
		is given by the output of Algorithm~\ref{alg:wmpava} on $z=(x_1,\dots,x_n)$ and $w=(\sigma_1^{-2},\dots,\sigma_n^{-2})$.
	\end{corollary}
	\begin{proof}
		The optimization \eqref{hetero} is rewritten as
		\begin{align}
			\hat{\mu}_{\lambda} &= \argmin_{\mu}\sum_{i=1}^n \frac{(X_i-\mu_i)^2}{2 \sigma_i^2} + \lambda \sum_{i=1}^{n-1} (\mu_i-\mu_{i+1})_+ \\
			&= \argmin_{\mu}\sum_{i=1}^n \left( -\mu_i \frac{X_i}{\sigma_i^2} + \sigma_i^{-2} \frac{\mu_i^2}{2} \right) + \lambda \sum_{i=1}^{n-1} (\mu_i-\mu_{i+1})_+,
		\end{align}
		which has the form of \eqref{ggneariso} with $X_i$ replace by $\sigma_i^{-2} X_i$ and $\theta_i=\mu_i$, $\psi(\theta)=\theta^2/2$ and $w_i=\sigma_i^{-2}$.
		Thus, from Theorem~\ref{th:weighted}, its solution is given by the output of Algorithm~\ref{alg:wmpava} on $z=(\sigma_1^{-2} x_1/\sigma_1^{-2},\dots,\sigma_n^{-2} x_n/\sigma_n^{-2})=(x_1,\dots,x_n)$ and $w=(\sigma_1^{-2},\dots,\sigma_n^{-2})$.
	\end{proof}
	
	\begin{corollary}
		Let $X_i \sim {\rm Bi}(N_i,r_i)$ for $i=1,\dots,n$.
		Then, the estimator \eqref{ggneariso} is given by the output of Algorithm~\ref{alg:wmpava} on $z=(x_1/N_1,\dots,x_n/N_n)$ and $w=(N_1,\dots,N_n)$. 
	\end{corollary}
	
	\begin{corollary}
		Let $X_i \sim {\rm Po}(\lambda_i)$ for $i=1,\dots,n$.
		Then, the estimator \eqref{ggneariso} is given by the output of Algorithm~\ref{alg:wmpava} on $z=(x_1,\dots,x_n)$ and $w=(1,\dots,1)$. 
	\end{corollary}
	
	\begin{corollary}
		Let $X_i \sim \sigma_i^2 \chi^2(d_i)$ for $i=1,\dots,n$.
		Then, the estimator \eqref{ggneariso} is given by the output of Algorithm~\ref{alg:wmpava} on $z=(x_1/d_1,\dots,x_n/d_n)$ and $w=(d_1/2,\dots,d_n/2)$. 
	\end{corollary}
	
	In some situations, we may have bound constraints on $\theta$ (e.g. Section~\ref{sec:ODE}).
	Such cases can be solved by simply thresholding the original solution as follows.
	Its proof is given in Appendix.
	
	\begin{proposition}\label{th:bounded}
		Let
		\begin{align}
			\hat{\theta}_{\lambda,\alpha,\beta} &= \argmin_{\theta \in [\alpha,\beta]^n} -\sum_{i=1}^n \log p(X_i \mid \theta_i) + \lambda \sum_{i=1}^{n-1} (\theta_i-\theta_{i+1})_+. \label{bounded}
		\end{align}
		Then, 
		\begin{align*}
			(\hat{\theta}_{\lambda,\alpha,\beta})_i = \min( \max ((\hat{\theta}_{\lambda})_i, \alpha), \beta)
		\end{align*}
		for $i=1,\dots,n$, where $\hat{\theta}_{\lambda}$ is given by \eqref{ggneariso}.
	\end{proposition}
	
	\subsection{Information criterion}
	In practice, the selection of the regularization parameter $\lambda$ is a crucial issue like other regularized estimators such as LASSO.
	Here, we propose a method for selecting $\lambda$ based on data by using an information criterion \citep{anderson2004model,konishi2008information}.
	
	First, we recall the following result by \cite{tibshirani2011nearly} for nearly isotonic regression \eqref{org_neariso}.
	
	\begin{proposition}\citep{tibshirani2011nearly}\label{prop:df}
		For the nearly isotonic regression $\hat{\mu}_{\lambda}$ in \eqref{org_neariso}, let $K_{\lambda}$ be the number of joined pieces in $\hat{\mu}_{\lambda}$.
		Then, the quantity
		\begin{align*}
			\hat{C}_p (\lambda) = {\| \hat{\mu}_{\lambda}-X \|^2} + 2 \sigma^2 K_{\lambda} - n\sigma^2 
		\end{align*}
		is an unbiased estimate of the mean squared error of $\hat{\mu}_{\lambda}$:
		\begin{align*}
			{\rm E}_{\mu} [\hat{C}_p(\lambda)] = {\rm E}_{\mu} \left[ {\| \hat{\mu}_{\lambda}-\mu \|^2} \right].
		\end{align*}
	\end{proposition}
	
	Proposition~\ref{prop:df} indicates that the number of joined pieces $K_{\lambda}$ is an unbiased estimate of the degrees of freedom \citep{efron2004estimation} of the nearly isotonic regression \eqref{org_neariso}.
	Based on this result, \cite{tibshirani2011nearly} selected the regularization parameter $\lambda$ by minimizing $\hat{C}_p(\lambda)$ among the knots obtained from the modified PAVA.
	Note that a similar result on the degrees of freedom has been obtained for other estimators such as isotonic regression \citep{meyer2000degrees} and LASSO \citep{zou2007degrees}.
	In particular, \cite{zou2007degrees} showed that the number of nonzero regression coefficients is an unbiased estimate of the degrees of freedom for LASSO, and proposed to use it as the penalty term of AIC and BIC.
	
	Now, we propose an information criterion for the estimator \eqref{ggneariso}.
	Following the convention of information criteria \citep[Chapter~3]{konishi2008information}, we interpret each $X_i$ as the sufficient statistic $X_{i1}+\dots+X_{im_i}$ for $\theta_i$ from independent samples $X_{ij} \sim \tilde{p}_i(x_i \mid \theta_i)$ for $j=1,\dots,m_i$, and consider the asymptotics $m_i \to \infty$ for every $i$.
	For example, when each $X_i \sim {\rm Bi} (N_i,r_i)$ is a binomial random variable, $m_i$ is set to $N_i$ and each $X_{ij}$ is taken to be the Bernoulli random variable with success probability $r_i$.
	The asymptotics $m_i \to \infty$ corresponds to $N_i \to \infty$, $\lambda_i \to \infty$ and $a_i \to \infty$ in the binomial, Poisson and gamma models, respectively.
	Then, we consider prediction of $Y_i \sim p_i(y_i \mid \theta_i)$ for $i=1,\dots,n$ by using the estimator $\hat{\theta}_{\lambda}$ in \eqref{ggneariso}.
	The prediction error is evaluated by the Kullback--Leibler discrepancy defined as
	\begin{align*}
		D(\theta,\hat{\theta}_{\lambda}) = {\rm E}_{\theta} \left[ -\sum_{i=1}^n \log p_i(Y_i \mid (\hat{\theta}_{\lambda})_i) \right],
	\end{align*}
	which is equivalent to the Kullback--Leibler divergence between $p(y \mid \theta)$ and $p(y \mid \hat{\theta})$ up to an additive constant.
	By using the unbiased estimate of the degrees of freedom in Proposition~\ref{prop:df}, we adopt
	\begin{align}\label{AIC}
		{\rm AIC}(\lambda) = -2 \sum_{i=1}^n \log p_i(X_i \mid (\hat{\theta}_{\lambda})_i) + 2 K_{\lambda}
	\end{align}
	as an approximately unbiased estimator of the expected Kullback--Leibler discrepancy.
	From the same argument with the usual derivation of information criteria, the bias evaluation reduces to that for the Gaussian model up to $O (m_i^{-1})$ as $m_i \to \infty$  \citep{konishi2008information}.
	Therefore, by using Proposition~\ref{prop:df},
	\[
	{\rm E}_{\theta} [{\rm AIC}(\lambda)] = 2 {\rm E}_{\theta} [D(\theta,\hat{\theta}_{\lambda})] + O (m^{-1})
	\]
	as $m=\min_i m_i \to \infty$.
	Thus, we select the regularization parameter by minimizing ${\rm AIC}(\lambda)$ among knots:
	\begin{align*}
		\hat{\lambda} = \lambda_{\hat{k}}, \quad \hat{k} = \argmin_k {\rm AIC}(\lambda_k).
	\end{align*}
	We will show the validity of this method by simulation in Section~\ref{sec:simulation}.
	
	\begin{remark}
		\cite{ninomiya2016aic} derived an information criterion for $l_1$-regularized estimators in generalized linear models, which can be viewed as an extension of the result of \cite{zou2007degrees} on the degrees of freedom of LASSO in Gaussian linear models.
		Their criterion is an approximately unbiased estimator of the expected Kullback--Leibler discrepancy and its bias correction term does not admit a simple closed-form solution, which is similar to TIC and GIC \citep{konishi2008information}.
		However, their simulation results imply that the bias correction term can be approximated well by twice the number of non-zero regression coefficients, which is shown to be an unbiased estimate of the degrees of freedom in the case of Gaussian linear models \citep{zou2007degrees}, especially when the sample size is large.
		Similarly, our simulation results below indicate that the unbiased estimate of the degrees of freedom in Gaussian nearly isotonic regression (Proposition~\ref{prop:df}) works well as a bias correction term of information criterion as long as the distribution is not very far from Gaussian.
		It is an interesting future problem to develop a more rigorous theory for this.
	\end{remark}
	
	\section{Simulation results}\label{sec:simulation}
	We check the performance of the proposed method for the binomial distribution.
	For $i=1,\dots,100$, let $X_i$ be a sample from the binomial distribution with $N_i$ trials and success probability $r_i$, where $r_1,\dots,r_{100}$ is a piecewise monotone sequence defined by
	\begin{align*}
		r_i = \begin{cases} 0.2 + 0.6 \cdot \frac{i-1}{49} & (i=1,\dots,50) \\ 0.2 + 0.6 \cdot \frac{i-51}{49} & (i=51,\dots,100) \end{cases}.
	\end{align*}
	We apply the proposed method to estimate $r_1,\dots,r_{100}$ from $X_1,\dots,X_{100}$.
	
	First, we set $N_i=10$ for $i=1,\dots,100$.
	Figure~\ref{fig:binomial} shows $\hat{r}_{\lambda}$ for several knot values of $\lambda$.
	Similarly to the original nearly isotonic regression, the estimate is piecewise monotone and the number of joined pieces decreases as $\lambda$ increases.
	In this case, $\hat{r}_{\lambda}$ becomes monotone at the final knot $\lambda = 50.5$ and it coincides with the result of the proposed method.
	Figure~\ref{fig:binomial_AIC} plots ${\rm AIC}(\lambda)$ with respect to $\lambda$.
	It takes minimum at $\hat{\lambda}=6.04$, which corresponds to the third panel of Figure~\ref{fig:binomial}.
	In this way, the proposed information criterion enables us to detect change-points in the parameter sequence of exponential families in a data-driven manner.
	
	\input{binomial.tex}
	
	\begin{figure}[htbp]
\centering
\begin{tikzpicture}
\tikzstyle{every node}=[]
\begin{axis}[
width=7cm,
xlabel={$\lambda$},
ylabel={${\rm AIC}(\lambda)$},
]
\addplot[very thick, color = blue, opacity=1,
] table {
0.5 412.952678644785
0.666666666666666 416.226416792188
0.666666666666667 414.226416792188
0.75 415.06984896609
0.999999999999999 418.665706794061
1 396.665706794061
1 388.665706794061
1 386.665706794061
1 384.665706794061
1.33333333333333 397.869157098019
1.33333333333333 395.869157098019
1.5 399.519049854198
1.66666666666667 403.621189899287
1.66666666666667 401.621189899288
1.66666666666667 399.621189899288
1.66666666666667 397.621189899288
1.66666666666668 395.621189899288
1.8 398.316614654518
1.85714285714286 398.287116775056
2 401.303287896656
2 389.303287896656
2 385.303287896656
2 387.303287896656
2 385.303287896656
2 385.303287896656
2.00000000000001 383.303287896656
2.00000000000001 381.303287896656
2.00000000000001 379.303287896656
2.00000000000002 379.303287896656
2.00000000000002 377.303287896656
2.25 380.057626441915
2.33333333333333 381.539253198299
2.42857142857143 376.816084878091
2.5 375.692959346052
2.57142857142857 374.495379046065
3 376.554590135697
3 374.554590135697
3 372.554590135698
3.36363636363637 373.578564587419
3.4 371.861416833865
3.4 369.861416833865
3.46153846153847 368.171664107493
4 366.275015445552
4.66666666666667 368.920526163425
4.88888888888889 367.413010768792
5.00000000000001 365.772691548539
5.42857142857143 364.989822427083
6.04761904761904 364.626161974384
10.5 372.297926134714
11.7142857142857 373.226409706347
14 374.195094351666
33.0909090909091 413.814185589532
50.5283018867925 454.086637817417
};
\end{axis}
\end{tikzpicture} 
\caption{AIC for the binomial distribution ($N=10$).}
\label{fig:binomial_AIC}
\end{figure}
	
	Next, we set $N_i=N$ for $i=1,\dots,100$ with $N \in \{ 10,20,30,50 \}$.
	Figure~\ref{fig:binomial_bias} plots ${\rm E}_{\theta}[{\rm AIC}(\lambda)]$ and $2 {\rm E}_{\theta} [D(\theta,\hat{\theta}_{\lambda})]$ with respect to $\lambda$ for each value of $N$, where we used 10000 repetitions.
	They show similar behaviors and take minimum at similar values of $\lambda$.
	Thus, the proposed information criterion is approximately unbiased.
	The absolute bias $|{\rm E}_{\theta}[{\rm AIC}(\lambda)]-2 {\rm E}_{\theta} [D(\theta,\hat{\theta}_{\lambda})]|$ decreases as $N$ increases, which is compatible with the fact that the binomial distribution becomes closer to the normal distribution for larger $N$.
	
	\input{binomial_bias.tex}
	
	Finally, we examine the case where the number of trials is not constant:
	\begin{align*}
		N_i = \begin{cases} 30 & (i=1,4,\dots,100) \\ 40 & (i=2,5,\dots,98) \\ 50 & (i=3,6,\dots,99) \end{cases}.
	\end{align*}
	Figure~\ref{fig:binomial_bias_weighted} plots ${\rm E}_{\theta}[{\rm AIC}(\lambda)]$ and $2 {\rm E}_{\theta} [D(\theta,\hat{\theta}_{\lambda})]$ with respect to $\lambda$, where we used 10000 repetitions.
	The bias of the proposed information criterion is sufficiently small.
	Thus, this criterion works well for determining the regularization parameter $\lambda$ even when the number of trials is heterogeneous among samples.
	
	\input{binomial_bias_weighted.tex}
	
	See Appendix for a similar experiment on the chi-square distribution.
	
	\section{Applications}\label{sec:application}
	
	\subsection{Spectrum estimation}
	Spectrum analysis is an important step in time series analysis that reveals periodicities in time series data \citep{brillinger2001time,brockwell2009time}.
	Specifically, the spectral density function of a Gaussian stationary time series $X=(X_t \mid t \in \mathbb{Z})$ is defined as
	\begin{align*}
		p(f) = \sum_{k=-\infty}^{\infty} C_k \exp(-2 \pi i k f), \quad -\frac{1}{2} \leq f \leq \frac{1}{2},
	\end{align*}
	where $C_k={\rm Cov}[X_t X_{t+k}]$ is the autocovariance.
	Let
	\begin{align}
		p_j = \frac{1}{2 \pi T} \left| \sum_{t=1}^T x_t \exp \left( - \frac{2 \pi i j t}{T} \right) \right|^2, \quad j=1,\dots,\frac{T}{2},
	\end{align}
	be the periodogram of the observation $x_1,\dots,x_T$.
	Then, from the theory of the Whittle likelihood \citep{whittle1953estimation}, the distribution of the periodogram is well approximated by the independent chi-square distributions:
	\begin{align}
		p_j \sim \frac{p(j/T)}{2} \chi^2(2),  \quad j=1,\dots,\frac{T}{2}.
	\end{align}
	Based on this property, many methods have been developed to estimate the spectral density function by smoothing the periodogram \citep{brillinger2001time}.
	
	The spectral density function of real time series data often tends to be decreasing \citep{anevski2011monotone} such as $1/f$ fluctuation (power law), possibly with a few peaks corresponding to characteristic periodicities or dominant frequencies.
	Thus, the proposed method is considered to be useful for estimating such nearly monotone spectral density functions.
	Figure~\ref{fig:sunspot} shows the result on the Wolfer sunspot data, which is the annual number of recorded sunspots on the sun's surface for the period 1770-1869 \citep{brockwell2009time}.
	Note that the result is shown in log-scale following the convention of spectrum analysis, whereas we applied the proposed method to the raw periodogram.
	This figure indicates one dominant frequency around 0.1 cycle per year.
	This frequency corresponds well to the well-known characteristic period of approximately 11 years in the sunspot number.
	
	\begin{figure}[htbp]
\centering
\begin{tikzpicture}
\tikzstyle{every node}=[]
\begin{axis}[
width=7cm,
xlabel={cycle per year},
ylabel={log-spectrum}
]
\addplot[only marks,mark options={scale=0.3},
filter discard warning=false, unbounded coords=discard,
] table {
0.01020408163265306 2.8002417311409378
0.02040816326530612 2.534371158360778
0.030612244897959183 -1.781312494089522
0.04081632653061224 -1.2281844046166954
0.05102040816326531 -1.7918206115260376
0.061224489795918366 0.44691198291293543
0.07142857142857142 2.626191293538474
0.08163265306122448 2.3056500545167635
0.09183673469387756 3.2363554997315327
0.10204081632653061 1.98666611406473
0.11224489795918367 0.47697269621018584
0.12244897959183673 1.7602053697649525
0.1326530612244898 -0.6964906389780071
0.14285714285714285 0.4652137615797168
0.15306122448979592 -0.21187068436370077
0.16326530612244897 -0.510931992481446
0.17346938775510204 0.5957836571576057
0.1836734693877551 0.44203166895425744
0.19387755102040816 -1.183952110773851
0.20408163265306123 -2.89916575219991
0.21428571428571427 -0.031130560171614054
0.22448979591836735 -0.46356412749363396
0.23469387755102042 -0.691195573029122
0.24489795918367346 -0.717622821438716
0.25510204081632654 -0.7531290565772014
0.2653061224489796 -0.6339120566642374
0.2755102040816326 -1.898582401679784
0.2857142857142857 -3.5031768495969344
0.29591836734693877 -4.969799966417452
0.30612244897959184 -1.8537628839147162
0.3163265306122449 -1.1455827278434398
0.32653061224489793 -2.496942488037192
0.336734693877551 -3.501023419831666
0.3469387755102041 -1.8235220564884924
0.35714285714285715 -1.801981823061021
0.3673469387755102 -2.294977252865811
0.37755102040816324 -4.353389664961251
0.3877551020408163 -3.360830533301883
0.3979591836734694 -2.3467970992593257
0.40816326530612246 -3.072909131349413
0.41836734693877553 -2.7029092423585515
0.42857142857142855 -5.867249529152819
0.4387755102040816 -3.1955543369704076
0.4489795918367347 -3.150309447236216
0.45918367346938777 -3.256072246600163
0.46938775510204084 -1.6660215827382134
0.47959183673469385 -1.9721731732882677
0.4897959183673469 -2.4492010406963316
0.5 -4.728166420770917
};
\addplot[very thick, color = blue, opacity=1,
] table {
0.01020408163265306 2.8002417311409378
0.02040816326530612 2.534371158360778
0.030612244897959183 -0.09463378687804337
0.04081632653061224 -0.09463378687804337
0.05102040816326531 -0.09463378687804337
0.061224489795918366 0.44691198291293543
0.07142857142857142 2.4787093997231797
0.08163265306122448 2.4787093997231797
0.09183673469387756 3.1501520276671373
0.10204081632653061 1.98666611406473
0.11224489795918367 1.3116813513472863
0.12244897959183673 1.3116813513472863
0.1326530612244898 0.13540586060523832
0.14285714285714285 0.13540586060523832
0.15306122448979592 0.13540586060523832
0.16326530612244897 0.13540586060523832
0.17346938775510204 0.13540586060523832
0.1836734693877551 0.13540586060523832
0.19387755102040816 -0.7057732580403493
0.20408163265306123 -0.7057732580403493
0.21428571428571427 -0.7057732580403493
0.22448979591836735 -0.7057732580403493
0.23469387755102042 -0.7057732580403493
0.24489795918367346 -0.7057732580403493
0.25510204081632654 -0.7057732580403493
0.2653061224489796 -0.7057732580403493
0.2755102040816326 -1.898582401679784
0.2857142857142857 -2.056249983098172
0.29591836734693877 -2.056249983098172
0.30612244897959184 -2.056249983098172
0.3163265306122449 -2.056249983098172
0.32653061224489793 -2.209708409175487
0.336734693877551 -2.209708409175487
0.3469387755102041 -2.209708409175487
0.35714285714285715 -2.209708409175487
0.3673469387755102 -2.294977252865811
0.37755102040816324 -2.712801179755263
0.3877551020408163 -2.712801179755263
0.3979591836734694 -2.712801179755263
0.40816326530612246 -2.712801179755263
0.41836734693877553 -2.712801179755263
0.42857142857142855 -2.712801179755263
0.4387755102040816 -2.712801179755263
0.4489795918367347 -2.712801179755263
0.45918367346938777 -2.712801179755263
0.46938775510204084 -2.712801179755263
0.47959183673469385 -2.712801179755263
0.4897959183673469 -2.712801179755263
0.5 -4.728166420770917
};
\end{axis}
\end{tikzpicture}
\caption{Result on the Wolfer sunspot data. black: log-periodogram, blue: generalized nearly isotonic regression.}
\label{fig:sunspot}
\end{figure}
	
	Estimation of a monotone spectral density has been studied in \cite{anevski2011monotone}.
	They proposed two estimators given by the isotonic regression of the periodogram and log-periodogram.
	They derived their asymptotic distributions and showed that they are rate optimal.
	While the isotonic regression of the periodogram  has smaller asymptotic variance than that of the log-periodogram, the latter has the advantage of being applicable to both short-memory and long-memory processes.
	Note that these estimators were defined as the solutions of the least squares problems, not the maximizer of the Whittle likelihood.
	It is an interesting future work to extend the result of \cite{anevski2011monotone} to estimation of a piecewise monotone spectral density.
	
	\subsection{Causal inference}
	Regression discontinuity design (RDD) is a statistical method for causal inference in econometrics \citep[Chapter~4]{angrist2014mastering}.
	It focuses on natural experiment situations where the assignment of a treatment is determined by some threshold of a covariate.
	One example is a scholarship that is given to all students above a threshold grade.
	Then, the (local) treatment effect is estimated by taking the difference of the average outcomes of the treatment and control groups at the threshold, which are estimated by applying parametric or nonparametric regression to each group separately.
	
	Here, we explore a possibility of applying the proposed method to RDD.
	We use the minimum legal drinking age data, which is a well-known example of RDD \citep[Chapter~4]{angrist2014mastering}.
	This data consists of the number of fatalities (per one-hundred thousands) for several causes of death by age in month\footnote{\url{https://www.masteringmetrics.com/resources/}}.
	We applied the proposed method with the Poisson distribution to the number of fatalities induced by motor vehicle accidents in 19-23 years old.
	Since the mortality has decreasing trend as a whole, we employed the regularization term $(\theta_{i+1}-\theta_i)_+$ instead of $(\theta_i-\theta_{i+1})_+$.
	Figure~\ref{fig:rdd} shows the result.
	There is a sudden increase at 21 years old, which coincides with the minimum legal drinking age.
	Thus, it can be interpreted as the effect of drunk driving on the number of fatalities induced by motor vehicle accidents.
	In this way, the proposed method may be useful for RDD in some cases, especially when the threshold of treatment assignments is not known a priori and has to be estimated simultaneously with the treatment effect \citep{porter2015regression}.
	Note that this method is applicable to RDD with categorical outcomes as well \citep{xu2017regression}.
	Recently, \cite{babii2021isotonic} proposed an application of isotonic regression to RDD.
	
	Recently, RDD has been applied to situations where the treatment assignment is based on geographic boundaries \citep{keele2015geographic} and it is called the spatial RDD.
	From our viewpoint, some of spatial RDD can be viewed as piecewise monotone estimation under partial orders induced from the geographic boundaries.  
	It is an interesting future work to extend the proposed method to such partially ordered cases.
	Note that the isotonic regression is applicable to partial orders as well \citep[Chapter~1]{rwd88}, such as multi-dimensional lattices \citep{anevski2018asymptotic,beran2010least} and graphs \citep{minami2020estimating}.
	
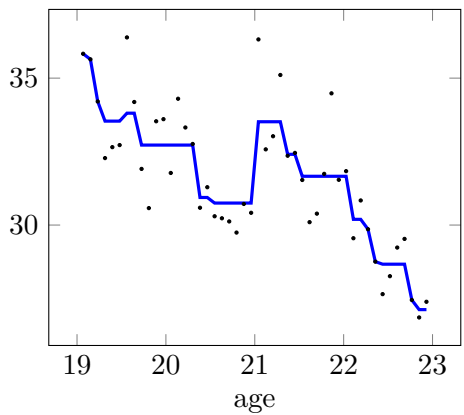
	\begin{figure}[htbp]
\centering
\begin{tikzpicture}
\tikzstyle{every node}=[]
\begin{axis}[
xlabel={age},
width=7cm
]
\addplot[only marks,mark options={scale=0.3},
filter discard warning=false, unbounded coords=discard,
] table {
19.068492889 35.82933
19.150684357 35.63926
19.232875824 34.20565
19.315069199 32.27896
19.397260666 32.65097
19.479452133 32.72144
19.5616436 36.3852
19.643835068 34.18793
19.726026535 31.91047
19.80821991 30.57683
19.890411377 33.53165
19.972602844 33.60344
20.054794312 31.77282
20.136985779 34.29831
20.219177246 33.32011
20.301370621 32.7594
20.383562088 30.58905
20.465753555 31.28775
20.547945023 30.29893
20.63013649 30.23012
20.712327957 30.12258
20.794521332 29.74465
20.876712799 30.71792
20.958904266 30.41714
21.041095734 36.31681
21.123287201 32.5758
21.205478668 33.02229
21.287672043 35.10687
21.36986351 32.3587
21.452054977 32.45526
21.534246445 31.53369
21.616437912 30.09868
21.698629379 30.38829
21.780822754 31.74226
21.863014221 34.48193
21.945205688 31.54001
22.027397156 31.83429
22.109588623 29.55155
22.19178009 30.83699
22.273973465 29.85789
22.356164932 28.75559
22.4383564 27.64927
22.520547867 28.26014
22.602739334 29.23415
22.684930801 29.52876
22.767124176 27.44976
22.849315643 26.85506
22.931507111 27.3893
};
\addplot[very thick, color = blue, opacity=1,
] table {
19.068492889 35.82933
19.150684357 35.63926
19.232875824 34.20565
19.315069199 33.536781904761895
19.397260666 33.536781904761895
19.479452133 33.536781904761895
19.5616436 33.80707714285714
19.643835068 33.80707714285714
19.726026535 32.72162875
19.80821991 32.72162875
19.890411377 32.72162875
19.972602844 32.72162875
20.054794312 32.72162875
20.136985779 32.72162875
20.219177246 32.72162875
20.301370621 32.72162875
20.383562088 30.9384
20.465753555 30.9384
20.547945023 30.748385952380943
20.63013649 30.748385952380943
20.712327957 30.748385952380943
20.794521332 30.748385952380943
20.876712799 30.748385952380943
20.958904266 30.748385952380943
21.041095734 33.515698571428565
21.123287201 33.515698571428565
21.205478668 33.515698571428565
21.287672043 33.515698571428565
21.36986351 32.406980000000004
21.452054977 32.406980000000004
21.534246445 31.659878571428575
21.616437912 31.659878571428575
21.698629379 31.659878571428575
21.780822754 31.659878571428575
21.863014221 31.659878571428575
21.945205688 31.659878571428575
22.027397156 31.659878571428575
22.109588623 30.19427
22.19178009 30.19427
22.273973465 29.85789
22.356164932 28.75559
22.4383564 28.66808
22.520547867 28.66808
22.602739334 28.66808
22.684930801 28.66808
22.767124176 27.44976
22.849315643 27.12218
22.931507111 27.12218
};
\end{axis}
\end{tikzpicture}
\caption{Result on the minimum legal drinking age data. black: data, blue: generalized nearly isotonic regression.}
\label{fig:rdd}
\end{figure}
	
	\subsection{Discretization error quantification of ODE solvers}\label{sec:ODE}
	
	Numerical integration of ordinary differential equations (ODEs) plays an essential role in many research fields.
	It is used not only for the future prediction but also for estimating the past states and/or system parameters in data assimilation. 
	The theory of numerical analysis tells us how the error induced by the discretization (e.g., Euler, Runge--Kutta) propagates, but standard discussion focuses on asymptotic behavior as the discretization stepsize goes to zero~\citep{hn93,hw96}. 
	The expense of sufficiently accurate numerical integration is often prohibitive. 
	Thus, in such cases, quantifying the reliability of numerical integration is essential. 
	In the last few years, several approaches to quantifying the discretization error of ODE solvers have been developed (see, for example, \cite{ag20,cg17,cc16,cosg19,ls19,ocag19,tksh18,ts21}).
	Here, we apply the proposed method to discretization error quantification of ODE solvers.
	
	Consider the ordinary differential equation
	\begin{equation}
		\label{eq:ivp}
		\frac{\rmd}{\rmd t} x(t) = f(x(t)), \quad x(0) = x_0\in\bbR^{m},
	\end{equation}
	where the vector field $f:\bbR^m \to \bbR^m$ is assumed to be sufficiently differentiable.
	For time points $t_1,\dots,t_n$, let $x_i$ be an approximation to $x(t_i)$ obtained by applying an ODE solver such as Runge--Kutta to \eqref{eq:ivp}.
	Also, we assume that we have noisy observations $y_1,\dots,y_n$ of $x(t_1),\dots,x(t_n)$:
	\begin{equation}
		\label{obs_model} 
		y_i = x_k(t_i) + \varepsilon_i, \quad \varepsilon_i \sim \rmN (0,\gamma^2), \quad i = 1,\dots,n,
	\end{equation}
	where we focus on a specific element $x_k$ to simplify the notation.
	We consider quantifying the discretization error $\xi_i := (x_i)_k - x_k(t_i)$ for $i=1,\dots,n$ based on $x_1,\dots,x_n$ and $y_1,\dots,y_n$.
	Note that we do not necessarily intend to estimate the discretization error as precisely as possible;
	instead, we aim to capture the scale of the discretization error and its qualitative behavior such as periodicity.
	
	Building on our previous study~\citep{mm21b}, we model the discretization error as independent Gaussian random variables:
	\begin{equation}
		\label{de_model}
		\xi_i \sim \mathrm N (0,\sigma_i^2), \quad i = 1,\dots,n,
	\end{equation}
	where the variance $\sigma_i^2$ quantifies the magnitude of $\xi_i$.
	By substituting \eqref{de_model} into \eqref{obs_model}, we obtain
	\begin{equation}
		y_i = (x_i)_k + e_i, \quad e_i \sim \mathrm N (0,\gamma^2 + \sigma_i^2), \quad i = 1,\dots,n,
	\end{equation}
	where $e_i := - \xi_i + \varepsilon_i$.
	Thus, the square of the residual $r_i=y_i-(x_i)_k$ follows the chi-square distribution with one degree of freedom:
	\begin{equation}
		r_i^2 \sim (\gamma^2 + \sigma_i^2) \chi^2(1), \quad i=1,\dots,n.
	\end{equation}
	In the following, we introduce a block constraint on $\sigma_1^2,\dots,\sigma_n^2$ with block size $d \geq 1$:
	\[
	\sigma_{(j-1)d+1}^2=\cdots=\sigma_{jd}^2=\widetilde{\sigma}_j^2, \quad j=1,\dots,\frac{n}{d},
	\]
	where the block size $d$ controls the smoothness of $\sigma_1^2,\dots,\sigma_n^2$ and $n$ is assumed to be divisible by $d$ for simplicity\footnote{If $n$ is indivisible by $d$, we simply ignore the data for the remaining indices}.
	Then, by putting ${s}_j=r_{(j-1)d+1}^2+\cdots+r_{jd}^2$, we have
	\begin{equation}
		s_j \sim (\gamma^2 + \widetilde{\sigma}_j^2) \chi^2(d), \quad j=1,\dots,\frac{n}{d}.
	\end{equation}
	We apply the proposed method to estimate $\widetilde{\sigma}_1^2,\dots,\widetilde{\sigma}_{n/d}^2$ from $s_1,\dots,s_{n/d}$, where we employ Theorem~\ref{th:bounded} to guarantee $\widetilde{\sigma}_j^2 \geq 0$ for $j=1,\dots,n/d$.
	Note that the sequence $\widetilde{\sigma}_1^2,\dots,\widetilde{\sigma}_{n/d}^2$ is expected to be piecewise monotone increasing, since the discretization error basically accumulates in every step of numerical integration, with possible drops if the ODE has periodicity (see Figure~\ref{fig:FN_solutions1}).
	From simulation results in Appendix, $d \geq 3$ is recommended to avoid large bias of AIC.
	This method can be viewed as an extension of our previous approach with the generalized isotonic regression \citep{mm21b}.

	The rest of the subsection checks how the above formulation works for quantifying the discretization error of ODE solvers.
	The idea of the proposed method leads to an intuition that the formulation suits a problem for which the discretization error gets large as time passes but exhibits periodic nature locally.
	Thus, we employ the FitzHugh--Nagumo (FN) model \citep{f61,nay62}:
	\begin{equation} \label{eq:FN}
		\frac{\rmd V}{\rmd t}=c\bigg( V-\cfrac{V^3}{3}+R \bigg), \quad \frac{\rmd R}{\rmd t}=-\cfrac{1}{c} \paren*{V-a+bR}
	\end{equation}
	as a toy problem.
	Since the solution to the FN model is almost periodic, the discretization error also varies periodically as long as the numerical solution is stable and captures the periodic nature.
	
	We set the initial state and parameters to $V(0)=-1$, $R(0)=1$ and $(a,b,c) = (0.2,0.2,3.0)$.
	We apply the explicit Euler method with the step size $\Delta t = 0.025$ to \eqref{eq:FN}, and compare the numerical solution with the exact solution in Figure~\ref{fig:FN_solutions1}.
	It is observed that while the numerical approximations well capture the periodicity of the exact flow in a qualitative manner, its phase speed is slower than the exact flow, and the difference between the exact and numerical flows becomes significant as time passes.
	
	\input{fig_FN_solutions1}
	
	\begin{remark}
		Undoubtedly, it is easy to obtain much more accurate numerical solutions to the FN model.
		Nevertheless, we even employ the explicit Euler method with a relatively large step size as an example for which 
		sufficiently accurate numerical integration is hard to attain.
	\end{remark}

	Figure~\ref{fig:FN_V_d3} shows the result of discretization error quantification on $V$, where $V$ is observed with observation noise variance $0.01$ at $t_i = (i-1)h$ with $i=201,202,\dots,1200$ and $h=0.05$ (i.e., $V$ is observed for $t\in[10,60]$) and $d=3$.
	The top panel plots ${\rm AIC}(\lambda)$ with respect to $\lambda$.
	In this case, ${\rm AIC}(\lambda)$ is minimized at $\lambda_{\rm opt} = 0.1134$.
	The bottom panel plots the estimated discretization error ${\sigma}_i$ with the actual error $|V_i-V(t_i)|$, where the result of the generalized isotonic regression (i.e., sufficiently large $\lambda$) is also shown for comparison.
	It indicates that the proposed method with $\lambda_{\rm opt}$ captures the fluctuation of the discretization error in a more conformable manner than generalized isotonic regression. 
	We conducted similar experiments for $d=5,10$ and obtained almost the same discretization error quantification results.
	
	Figure~\ref{fig:FN_R_d3} shows the result for $R$, where the observation noise variance was set to $0.004$.
	The discussion for $V$ remains valid for $R$, although the error behavior for $R$ is different from that for $V$.
	For $V$, the error gets large moderately and then decreases quite sharply; for $R$, the error decreases moderately after a sharp increase.
	
	In summary, the proposed method can capture the periodicity and scale of the actual discretization error well compared with the previous one using the generalized isotonic regression.
	The new method seems beneficial in that, for example, we may be able to understand how the error propagates in a more accurate way and further detect recovery of the numerical reliability.
	This method is expected to be useful in the inverse problem framework, and we leave further discussions to our future work.
	
	\input{FN_V_d3}
	
	\input{FN_R_d3}
	
	\section{Conclusion}\label{sec:conclusion}
	In this study, we extended nearly isotonic regression to general one-parameter exponential families such as binomial, Poisson and chi-square.
	We developed an efficient algorithm based on the modified PAVA and provided a method for selecting the regularization parameter by using an information criterion.
	Simulation results demonstrated that the proposed method detects change-points in piecewise monotone parameter sequences in a data-driven manner.
	We presented applications to spectrum estimation, causal inference and discretization error quantification of ODE solvers.
	
	While we focused on simply ordered cases in this study, isotonic regression is also applicable to partially ordered cases \citep[Chapter~1]{rwd88}.
	Recent studies considered multi-dimensional lattices \citep{anevski2018asymptotic,beran2010least} and graphs \citep{minami2020estimating}.
	It is an interesting future work to extend the proposed method to such settings.
	Such a generalization may be applicable to spatial regression discontinuity design as well as discretization error quantification of PDE solvers, which would be useful for reliable simulation as well as large-scale data assimilation.
	
	We proposed an information criterion for selecting the regularization parameter based on a rather heuristic argument.
	Although it works practically well as long as the model is not very far from Gaussian, the bias is non-negligible in several cases such as the chi-square with a few degrees of freedom.
	It is a future problem to derive a more accurate information criterion like the one in \cite{ninomiya2016aic}.
	Note that the number of parameters grows with the sample size here, and thus the usual argument of Akaike information criterion is not directly applicable.
	Derivation of risk bounds like the one in \cite{minami2020estimating} is another interesting direction for future work.
	
	\section*{Acknowledgements}
We thank Yuya Shimizu and Koki Fusejima for helpful comments.
We thank Grace Chen for finding a bug of our code.
Takeru Matsuda was supported by JSPS KAKENHI Grant Numbers 19K20220, 21H05205 and 22K17865, and JST Moonshot Grant Number JPMJMS2024.
Yuto Miyatake was supported by JSPS KAKENHI Grant Numbers 20H01822, 20H00581 and 21K18301, and JST ACT-I Grant Number JPMJPR18US.

\bibliographystyle{chicago}      
\bibliography{references}   

\begin{thebibliography}{}

\bibitem[\protect\citeauthoryear{Abdulle and Garegnani}{Abdulle and Garegnani}{2020}]{ag20}
Abdulle, A. and G.~Garegnani (2020).
\newblock Random time step probabilistic methods for uncertainty quantification in chaotic and geometric numerical integration.
\newblock {\em Stat. Comput.\/}~{\em 30}, 907--932.

\bibitem[\protect\citeauthoryear{Amari}{Amari}{2016}]{amari2016information}
Amari, S.-i. (2016).
\newblock {\em Information Geometry and Its Applications}, Volume 194 of {\em Applied Mathematical Sciences}.
\newblock Springer, Tokyo.

\bibitem[\protect\citeauthoryear{Anevski and Pastukhov}{Anevski and Pastukhov}{2018}]{anevski2018asymptotic}
Anevski, D. and V.~Pastukhov (2018).
\newblock The asymptotic distribution of the isotonic regression estimator over a general countable pre-ordered set.
\newblock {\em Electronic Journal of Statistics\/}~{\em 12\/}(2), 4180--4208.

\bibitem[\protect\citeauthoryear{Anevski and Soulier}{Anevski and Soulier}{2011}]{anevski2011monotone}
Anevski, D. and P.~Soulier (2011).
\newblock Monotone spectral density estimation.
\newblock {\em Ann. Statist.\/}~{\em 39\/}(1), 418--438.

\bibitem[\protect\citeauthoryear{Angrist and Pischke}{Angrist and Pischke}{2014}]{angrist2014mastering}
Angrist, J.~D. and J.-S. Pischke (2014).
\newblock {\em Mastering 'Metrics: The Path From Cause to Effect}.
\newblock Princeton University Press.

\bibitem[\protect\citeauthoryear{Babii and Kumar}{Babii and Kumar}{2021}]{babii2021isotonic}
Babii, A. and R.~Kumar (2021).
\newblock Isotonic regression discontinuity designs.
\newblock {\em Journal of Econometrics\/}.

\bibitem[\protect\citeauthoryear{Barlow, Bartholomew, Bremner, and Brunk}{Barlow et~al.}{1972}]{bb72}
Barlow, R.~E., D.~J. Bartholomew, J.~M. Bremner, and H.~D. Brunk (1972).
\newblock {\em Statistical {I}nference {U}nder {O}rder {R}estrictions. {T}he {T}heory and {A}pplication of {I}sotonic {R}egression}.
\newblock John Wiley \& Sons, London-New York-Sydney.

\bibitem[\protect\citeauthoryear{Bellec}{Bellec}{2018}]{bellec2018sharp}
Bellec, P.~C. (2018).
\newblock Sharp oracle inequalities for least squares estimators in shape restricted regression.
\newblock {\em Ann. Statist.\/}~{\em 46\/}(2), 745--780.

\bibitem[\protect\citeauthoryear{Beran and D{\"u}mbgen}{Beran and D{\"u}mbgen}{2010}]{beran2010least}
Beran, R. and L.~D{\"u}mbgen (2010).
\newblock Least squares and shrinkage estimation under bimonotonicity constraints.
\newblock {\em Statistics and computing\/}~{\em 20\/}(2), 177--189.

\bibitem[\protect\citeauthoryear{Bertsekas}{Bertsekas}{1997}]{bertsekas1997nonlinear}
Bertsekas, D.~P. (1997).
\newblock Nonlinear programming.
\newblock {\em Journal of the Operational Research Society\/}~{\em 48\/}(3), 334--334.

\bibitem[\protect\citeauthoryear{Boyd and Vandenberghe}{Boyd and Vandenberghe}{2004}]{bv04}
Boyd, S. and L.~Vandenberghe (2004).
\newblock {\em Convex {O}ptimization}.
\newblock Cambridge University Press.

\bibitem[\protect\citeauthoryear{Brillinger}{Brillinger}{2001}]{brillinger2001time}
Brillinger, D.~R. (2001).
\newblock {\em Time Series: Data Analysis and Theory}.
\newblock SIAM, Philadelphia, PA.

\bibitem[\protect\citeauthoryear{Brockwell and Davis}{Brockwell and Davis}{2009}]{brockwell2009time}
Brockwell, P.~J. and R.~A. Davis (2009).
\newblock {\em Time Series: Theory and Methods\/} (Second ed.).
\newblock Springer, New York.

\bibitem[\protect\citeauthoryear{Burnham and Anderson}{Burnham and Anderson}{2002}]{anderson2004model}
Burnham, K.~P. and D.~R. Anderson (2002).
\newblock {\em Model selection and multi-model inference}.
\newblock Springer, New York.

\bibitem[\protect\citeauthoryear{Chkrebtii, Campbell, Calderhead, and Girolami}{Chkrebtii et~al.}{2016}]{cc16}
Chkrebtii, O.~A., D.~A. Campbell, B.~Calderhead, and M.~A. Girolami (2016).
\newblock Bayesian solution uncertainty quantification for differential equations.
\newblock {\em Bayesian Anal.\/}~{\em 11\/}(4), 1239--1267.

\bibitem[\protect\citeauthoryear{Cockayne, Oates, Sullivan, and Girolami}{Cockayne et~al.}{2019}]{cosg19}
Cockayne, J., C.~J. Oates, T.~Sullivan, and M.~Girolami (2019).
\newblock Bayesian probabilistic numerical methods.
\newblock {\em SIAM Rev.\/}~{\em 61}, 756--789.

\bibitem[\protect\citeauthoryear{Conrad, Girolami, S\"{a}rkk\"{a}, Stuart, and Zygalakis}{Conrad et~al.}{2017}]{cg17}
Conrad, P.~R., M.~Girolami, S.~S\"{a}rkk\"{a}, A.~Stuart, and K.~Zygalakis (2017).
\newblock Statistical analysis of differential equations: introducing probability measures on numerical solutions.
\newblock {\em Stat. Comput.\/}~{\em 27\/}(4), 1065--1082.

\bibitem[\protect\citeauthoryear{Efron}{Efron}{2004}]{efron2004estimation}
Efron, B. (2004).
\newblock The estimation of prediction error: covariance penalties and cross-validation.
\newblock {\em J. Amer. Statist. Assoc.\/}~{\em 99\/}(467), 619--632.

\bibitem[\protect\citeauthoryear{Efron}{Efron}{2022}]{Efron}
Efron, B. (2022).
\newblock {\em Exponential Families in Theory and Practice}.
\newblock Cambridge University Press.

\bibitem[\protect\citeauthoryear{FitzHugh}{FitzHugh}{1961}]{f61}
FitzHugh, R. (1961).
\newblock Impulses and physiological states in models of nerve membrane.
\newblock {\em Biophys. J.\/}~{\em 1}, 445--466.

\bibitem[\protect\citeauthoryear{Friedman, Hastie, H{\"o}fling, and Tibshirani}{Friedman et~al.}{2007}]{friedman2007pathwise}
Friedman, J., T.~Hastie, H.~H{\"o}fling, and R.~Tibshirani (2007).
\newblock Pathwise coordinate optimization.

\bibitem[\protect\citeauthoryear{Groeneboom and Jongbloed}{Groeneboom and Jongbloed}{2014}]{groeneboom2014nonparametric}
Groeneboom, P. and G.~Jongbloed (2014).
\newblock {\em Nonparametric Estimation Under Shape Constraints}, Volume~38.
\newblock Cambridge University Press, New York.

\bibitem[\protect\citeauthoryear{Guntuboyina and Sen}{Guntuboyina and Sen}{2018}]{guntuboyina2018nonparametric}
Guntuboyina, A. and B.~Sen (2018).
\newblock Nonparametric shape-restricted regression.
\newblock {\em Statist. Sci.\/}~{\em 33\/}(4), 568--594.

\bibitem[\protect\citeauthoryear{Hairer, N{\o}rsett, and Wanner}{Hairer et~al.}{1993}]{hn93}
Hairer, E., S.~P. N{\o}rsett, and G.~Wanner (1993).
\newblock {\em Solving Ordinary Differential Equations {I}. Nonstiff Problems\/} (Second ed.).
\newblock Springer-Verlag, Berlin.

\bibitem[\protect\citeauthoryear{Hairer and Wanner}{Hairer and Wanner}{1996}]{hw96}
Hairer, E. and G.~Wanner (1996).
\newblock {\em Solving Ordinary Differential Equations {II}. Stiff and Differential-Algebraic Problems\/} (Second ed.).
\newblock Springer-Verlag, Berlin.

\bibitem[\protect\citeauthoryear{Han, Wang, Chatterjee, and Samworth}{Han et~al.}{2019}]{han2019isotonic}
Han, Q., T.~Wang, S.~Chatterjee, and R.~J. Samworth (2019).
\newblock Isotonic regression in general dimensions.
\newblock {\em Ann. Statist.\/}~{\em 47\/}(5), 2440--2471.

\bibitem[\protect\citeauthoryear{Keele and Titiunik}{Keele and Titiunik}{2015}]{keele2015geographic}
Keele, L.~J. and R.~Titiunik (2015).
\newblock Geographic boundaries as regression discontinuities.
\newblock {\em Political Analysis\/}~{\em 23\/}(1), 127--155.

\bibitem[\protect\citeauthoryear{Konishi and Kitagawa}{Konishi and Kitagawa}{2008}]{konishi2008information}
Konishi, S. and G.~Kitagawa (2008).
\newblock {\em Information Criteria and Statistical Modeling}.
\newblock Springer Series in Statistics. Springer, New York.

\bibitem[\protect\citeauthoryear{Lehmann and Casella}{Lehmann and Casella}{2006}]{lehmann2006theory}
Lehmann, E.~L. and G.~Casella (2006).
\newblock {\em Theory of point estimation}.
\newblock Springer Science \& Business Media.

\bibitem[\protect\citeauthoryear{Lie, Sullivan, and Stuart}{Lie et~al.}{2019}]{ls19}
Lie, H.~C., T.~J. Sullivan, and A.~Stuart (2019).
\newblock Strong convergence rates of probabilistic integrators for ordinary differential equations.
\newblock {\em Stat. Comput.\/}~{\em 29}, 1265--1283.

\bibitem[\protect\citeauthoryear{Matsuda and Miyatake}{Matsuda and Miyatake}{2021}]{mm21b}
Matsuda, T. and Y.~Miyatake (2021).
\newblock Estimation of ordinary differential equation models with discretization error quantification.
\newblock {\em SIAM/ASA J. Uncertain. Quantif.\/}~{\em 9\/}(1), 302--331.

\bibitem[\protect\citeauthoryear{Meyer and Woodroofe}{Meyer and Woodroofe}{2000}]{meyer2000degrees}
Meyer, M. and M.~Woodroofe (2000).
\newblock On the degrees of freedom in shape-restricted regression.
\newblock {\em Ann. Statist.\/}~{\em 28\/}(4), 1083--1104.

\bibitem[\protect\citeauthoryear{Minami}{Minami}{2020}]{minami2020estimating}
Minami, K. (2020).
\newblock Estimating piecewise monotone signals.
\newblock {\em Electron. J. Stat.\/}~{\em 14\/}(1), 1508--1576.

\bibitem[\protect\citeauthoryear{Nagumo, Arimoto, and Yoshizawa}{Nagumo et~al.}{1962}]{nay62}
Nagumo, J.~S., S.~Arimoto, and S.~Yoshizawa (1962).
\newblock An active pulse transmission line simulating a nerve axon.
\newblock {\em Proc. Inst. Radio Engrs\/}~{\em 50}, 2061--2070.

\bibitem[\protect\citeauthoryear{Ninomiya and Kawano}{Ninomiya and Kawano}{2016}]{ninomiya2016aic}
Ninomiya, Y. and S.~Kawano (2016).
\newblock Aic for the lasso in generalized linear models.
\newblock {\em Electron. J. Stat.\/}~{\em 10\/}(2), 2537--2560.

\bibitem[\protect\citeauthoryear{Oates, Cockayne, Aykroyd, and Girolami}{Oates et~al.}{2019}]{ocag19}
Oates, C.~J., J.~Cockayne, R.~G. Aykroyd, and M.~Girolami (2019).
\newblock Bayesian probabilistic numerical methods in time-dependent state estimation for industrial hydrocyclone equipment.
\newblock {\em J. Am. Stat. Assoc.\/}~{\em 114}, 1518--1531.

\bibitem[\protect\citeauthoryear{Porter and Yu}{Porter and Yu}{2015}]{porter2015regression}
Porter, J. and P.~Yu (2015).
\newblock Regression discontinuity designs with unknown discontinuity points: testing and estimation.
\newblock {\em J. Econometrics\/}~{\em 189\/}(1), 132--147.

\bibitem[\protect\citeauthoryear{Robertson, Wright, and Dykstra}{Robertson et~al.}{1988}]{rwd88}
Robertson, T., F.~T. Wright, and R.~L. Dykstra (1988).
\newblock {\em Order Restricted Statistical Inference}.
\newblock Wiley Series in Probability and Mathematical Statistics: Probability and Mathematical Statistics. John Wiley \& Sons, Ltd., Chichester.

\bibitem[\protect\citeauthoryear{Sibuya, Kawai, and Shida}{Sibuya et~al.}{1990}]{sibuya1990equipartition}
Sibuya, M., T.~Kawai, and K.~Shida (1990).
\newblock Equipartition of particles forming clusters by inelastic collisions.
\newblock {\em Phys. A\/}~{\em 167\/}(3), 676--689.

\bibitem[\protect\citeauthoryear{Tibshirani, Hoefling, and Tibshirani}{Tibshirani et~al.}{2011}]{tibshirani2011nearly}
Tibshirani, R.~J., H.~Hoefling, and R.~Tibshirani (2011).
\newblock Nearly-isotonic regression.
\newblock {\em Technometrics\/}~{\em 53\/}(1), 54--61.

\bibitem[\protect\citeauthoryear{Tronarp, Kersting, Särkkä, and Hennig}{Tronarp et~al.}{2019}]{tksh18}
Tronarp, F., H.~Kersting, S.~Särkkä, and P.~Hennig (2019).
\newblock Probabilistic solutions to ordinary differential equations as non-linear {B}ayesian filtering: A new perspective.
\newblock {\em Stat. Comput.\/}~{\em 29}, 1297--1315.

\bibitem[\protect\citeauthoryear{Tronarp, S\"{a}rkk\"{a}, and Hennig}{Tronarp et~al.}{2021}]{ts21}
Tronarp, F., S.~S\"{a}rkk\"{a}, and P.~Hennig (2021).
\newblock Bayesian {ODE} solvers: the maximum a posteriori estimate.
\newblock {\em Stat. Comput.\/}~{\em 31\/}(3), Paper No. 23, 18.

\bibitem[\protect\citeauthoryear{van Eeden}{van Eeden}{2006}]{vE06}
van Eeden, C. (2006).
\newblock {\em Restricted {P}arameter {S}pace {E}stimation {P}roblems}.
\newblock Springer, New York.

\bibitem[\protect\citeauthoryear{Whittle}{Whittle}{1953}]{whittle1953estimation}
Whittle, P. (1953).
\newblock Estimation and information in stationary time series.
\newblock {\em Ark. Mat.\/}~{\em 2\/}(5), 423--434.

\bibitem[\protect\citeauthoryear{Xu}{Xu}{2017}]{xu2017regression}
Xu, K.-L. (2017).
\newblock Regression discontinuity with categorical outcomes.
\newblock {\em J. Econometrics\/}~{\em 201\/}(1), 1--18.

\bibitem[\protect\citeauthoryear{Zou, Hastie, and Tibshirani}{Zou et~al.}{2007}]{zou2007degrees}
Zou, H., T.~Hastie, and R.~Tibshirani (2007).
\newblock On the ``degrees of freedom'' of the lasso.
\newblock {\em Ann. Statist.\/}~{\em 35\/}(5), 2173--2192.

\end{thebibliography}

\appendix

\section{Background}\label{sec:background}
\subsection{Order restricted MLE of normal means}\label{sec:isotonic}
As discussed in the Introduction, order restricted MLE of normal means is reduced to the isotonic regression problem \eqref{iso}.
Since this is a convex optimization over the closed set of points $(\mu_1,\dots,\mu_n)$ that satisfy $\mu_1 \leq \dots \leq \mu_n$, the maximum likelihood estimator uniquely exists.
Figure~\ref{fig:background} presents an example.
This problem is efficiently solved by the pool adjacent violators algorithm (PAVA) given in Algorithm~\ref{alg:pava}.
See Chapter~1 of \cite{rwd88} for details.

\begin{algoalgo}[Pool adjacent violators algorithm, PAVA]\label{alg:pava}
	\mbox{}
	\begin{itemize}
		\item Start with $K=n$ clusters $A_i=\{ i \}$ with values $y_{A_i}=x_i$ for $i=1,\dots,n$.
		
		\item Repeat:
		\begin{itemize}
			\item If $y_{A_{j-1}}>y_{A_j}$ for some $j$, then merge $A_j$ into $A_{j-1}$, set its value to $(|A_{j-1}| y_{A_{j-1}}+|A_j| y_{A_j})/(|A_{j-1}|+|A_j|)$, renumber $A_{j+1},\dots,A_K$ to $A_j,\dots,A_{K-1}$ and decrease $K$ by one.
		\end{itemize}
		
		\item Return $\hat{\mu}$ with $\hat{\mu}_i=y_{A_j}$ for $i \in A_j$
	\end{itemize}
\end{algoalgo}

\input{background_fig.tex}

\subsection{Piecewise monotone estimation of normal means}\label{sec:neariso}
As discussed in the Introduction, piecewise monotone estimation of normal means (nearly isotonic regression) is formulated as \eqref{org_neariso}.
Each of the regularization term $(\mu_i-\mu_{i+1})_+$ is piecewise linear and non-differentiable at $\mu_i=\mu_{i+1}$.
This property leads to $(\hat{\mu}_{\lambda})_i \leq (\hat{\mu}_{\lambda})_{i+1}$ for sufficiently large $\lambda$ in the same way that the $l_1$ regularization term provides a sparse solution in LASSO.
Figure~\ref{fig:background} plots this estimator with $\lambda=3.68$.
Compared to the solution of isotonic regression, this estimator successfully captures the drop of $\mu_i$ around $i=100$.
Note that nearly isotonic regression coincides with isotonic regression when the regularization parameter $\lambda$ is sufficiently large.
Recently, \cite{minami2020estimating} investigated the risk bound of nearly isotonic regression.

The nearly isotonic regression is efficiently solved by a modification of PAVA (Algorithm~\ref{alg:wmpava} in the next Section with $w_1=\dots=w_n=1$).
This algorithm outputs the regularization path by computing the set of critical points (knots) $\lambda_0=0,\lambda_1,\dots,\lambda_T$ and the estimate $\hat{\mu}_{\lambda_t}$ at each critical point.
Since the solution path is piecewise linear between the critical points, the solution for general $\lambda$ is readily obtained by linear interpolation.

\begin{remark}
	For isotonic regression, several algorithms other than PAVA have been developed, such as the minimum lower set algorithm \citep[Section~1.4]{rwd88}.
	It is an interesting future work to extend these algorithms to nearly isotonic regression.
\end{remark}

In practice, it is important to select an appropriate value of the regularization parameter $\lambda$ based on data.
For this aim, \cite{tibshirani2011nearly} derived an unbiased estimate of the degrees of freedom \citep{efron2004estimation} of nearly isotonic regression.
Here, we briefly review this result.
Suppose that we have an observation $X \sim {\rm N}_n(\mu,\sigma^2 I)$ and estimate $\mu$ by an estimator $\hat{\mu}=\hat{\mu}(X)$, where $\sigma^2$ is known.
From Stein's lemma, the mean squared error of $\hat{\mu}$ is given by
\begin{align*}
	{\rm E}_{\mu} \left[ {\| \hat{\mu}-\mu \|^2} \right] = {\rm E}_{\mu} \left[ {\| \hat{\mu}-X \|^2} \right] + 2 \sigma^2 {\rm df}_{\mu} (\hat{\mu}) - n \sigma^2,
\end{align*}
where 
\begin{align*}
	{\rm df}_{\mu}(\hat{\mu}) = {\rm E}_{\mu} \left[ \sum_{i=1}^n \frac{\partial \hat{\mu}_i}{\partial x_i} (X) \right]
\end{align*}
is called the degrees of freedom of $\hat{\mu}$.
For example, the degrees of freedom of a linear estimator $\hat{\mu}=AX$ do not depend on $\mu$ and is equal to ${\rm tr} A$.
In general, the degrees of freedom depend on $\mu$ and unbiased estimates of them have been derived, which can be used for the penalty term of model selection criteria such as Mallows' $C_p$, AIC and BIC.
For isotonic regression, \cite{meyer2000degrees} showed that the number of joined pieces is an unbiased estimate of the degrees of freedom.
For LASSO, \cite{zou2007degrees} showed that the number of nonzero regression coefficients is an unbiased estimate of the degrees of freedom.
\cite{tibshirani2011nearly} proved a similar result for nearly isotonic regression as follows.

\begin{proposition}\citep{tibshirani2011nearly}\label{prop:df}
	Let $K_{\lambda}$ be the number of joined pieces in $\hat{\mu}_{\lambda}$.
	Then,
	\[
	{\rm E}_{\mu} [K_{\lambda}] = {\rm df}_{\mu} (\hat{\mu}_{\lambda}).
	\]
\end{proposition}

Therefore, the quantity
\begin{align*}
	\hat{C}_p (\lambda) = {\| \hat{\mu}_{\lambda}-X \|^2} + 2 \sigma^2 K_{\lambda} - n\sigma^2 
\end{align*}
is an unbiased estimate of the mean squared error of $\hat{\mu}_{\lambda}$:
\begin{align*}
	{\rm E}_{\mu} [\hat{C}_p(\lambda)] = {\rm E}_{\mu} \left[ {\| \hat{\mu}_{\lambda}-\mu \|^2} \right].
\end{align*}
Thus, \cite{tibshirani2011nearly} selected the regularization parameter $\lambda$ by minimizing $\hat{C}_p(\lambda)$ among the knots: 
\begin{align}
	\hat{\lambda} = \lambda_{\hat{k}}, \quad \hat{k} = \argmin_k \hat{C}_p(\lambda_k). \label{lambda_hat}
\end{align}
The value of $\lambda$ in Figure~\ref{fig:background} was selected by this method.

\subsection{Order restricted MLE in one-parameter exponential families}\label{sec:exp}
Consider a one-parameter exponential family
\begin{align}\label{expfamily}
	p(x \mid \theta) = h(x) \exp \left( \theta x - \psi (\theta) \right),
\end{align}
where $\psi$ is a smooth convex function.
This class includes many standard distributions such as binomial, Poisson and gamma \citep{lehmann2006theory,Efron}.
The binomial distribution ${\rm Bi}(N,r)$ with $N$ (fixed) trials of success probability $r$ corresponds to $x \in \{ 0,1,\dots,N \}$, $h(x) = N!/(x! (N-x)!)$ and $\psi(\theta) = N \log (1+e^{\theta})$, where $r=e^{\theta}/(1+e^{\theta})$.
The Poisson distribution ${\rm Po}(\lambda)$ with mean $\lambda$ corresponds to $x \in \{ 0,1,\dots \}$, $h(x) = 1/(x!)$ and $\psi(\theta) = e^{\theta}$, where $\lambda=e^{\theta}$.
The gamma distribution ${\rm Ga}(a,b)$ with shape $a$ (fixed) and scale $b$ corresponds to $x \geq 0$, $h(x) = x^{a-1}/\Gamma(a)$ and $\psi(\theta) = -a \log (-\theta)$, where $b=-1/\theta$, and it reduces to the chi-square distribution $\chi^2(d)$ with $d$ degrees of freedom when $a=d/2$ and $b=2$.
Also, the normal distribution ${\rm N}(\theta,1)$ with mean $\theta$ and variance one corresponds to $x \in \mathbb{R}$, $h(x)=(2 \pi)^{-1/2} \exp (-x^2/2)$ and $\psi(\theta) = \theta^2/2$.

Exponential families have two canonical parametrizations called the natural parameter $\theta$ and the expectation parameter $\eta={\rm E}_{\theta}[X]$.
They are dual in the sense that they have one-to-one correspondence given by $\eta=\psi'(\theta)$, which is related to the Legendre transform of the convex function $\psi$.
This duality plays a central role in information geometry and $\theta$ and $\eta$ are called the e-coordinate and m-coordinate, respectively \citep{amari2016information}.
Note that the normal model is self-dual: $\theta=\eta$.
The relation $\eta=\psi'(\theta)$ appears in the derivation of the first moment from the moment generating function.
See (5.14) in \cite{lehmann2006theory}.

For one-parameter exponential families, maximum likelihood estimation under order constraints reduces to a problem called the generalized isotonic regression and it is efficiently solved by PAVA as well \citep[Section~1.5]{rwd88}.
Suppose that we have $n$ observations $X_i \sim p(x \mid \theta_i)$ for $i=1,\dots,n$ where $\theta_1 \leq \cdots \leq \theta_n$.
Then, the maximum likelihood estimate of $\theta$ under the order constraint is given by 
\begin{align*}
	\hat{\theta} &= \argmax_{\theta_1 \leq \cdots \leq \theta_n} \sum_{i=1}^n \log p(X_i \mid \theta_i) = \argmin_{\theta_1 \leq \cdots \leq \theta_n} \sum_{i=1}^n (- \theta_i X_i + \psi(\theta_i)).
\end{align*}
This constrained optimization is solved by PAVA as follows.

\begin{proposition}\citep[Theorem~1.5.2]{rwd88}
	Let $\hat{\eta}=(\hat{\eta}_1,\dots,\hat{\eta}_n)$ be the output of PAVA on the realization $(x_1,\dots,x_n)$ of $(X_1,\dots,X_n)$. 
	Then, the maximum likelihood estimate of $\theta$ is given by $\hat{\theta}=(\hat{\theta}_1,\dots,\hat{\theta}_n)$ where $\hat{\theta}_i=(\psi')^{-1}(\hat{\eta}_i)$ for $i=1,\dots,n$.
\end{proposition}

\section{Proof of Lemma 1}
\begin{proof}
	
	We follow a similar discussion to \cite{tibshirani2011nearly}.
	The KKT condition \citep[Section~5.5.3]{bv04} for \eqref{ggneariso} is
	\begin{equation}
		\label{eq:KKT}
		w_i \psi'(\hat{\theta}_{\lambda,i}) - X_i + \lambda (s_{\lambda,i} - s_{\lambda,i-1}) = 0 \quad \text{for }i=1,\dots,n,
	\end{equation}
	where 
	\begin{equation*}
		s_{\lambda,i} \begin{cases}
			= 1 & (\hat{\eta}_{\lambda,i}- \hat{\eta}_{\lambda,i+1} > 0) \\
			= 0 & (\hat{\eta}_{\lambda,i}- \hat{\eta}_{\lambda,i+1} < 0) \\
			\in [0,1] & (\hat{\eta}_{\lambda,i}- \hat{\eta}_{\lambda,i+1} = 0) 
		\end{cases}.
	\end{equation*}
	Suppose that 
	\begin{equation*}
		(\hat{\mu}_{\tilde{\lambda},j-1}\neq) 
		\hat{\mu}_{\tilde{\lambda},j} = \hat{\mu}_{\tilde{\lambda},j+1} = \cdots = 
		\hat{\mu}_{\tilde{\lambda},j+k} (\neq \hat{\mu}_{\tilde{\lambda},j+k+1})
	\end{equation*}
	for some $\lambda = \tilde{\lambda}(\geq 0)$.
	Then, we have $s_{\tilde{\lambda},j-1}, s_{\tilde{\lambda},j+k} \in \{0,1\}$ and these values remain constant as $\lambda$ increases as long as $\hat{\mu}_{\lambda,j-1}\neq
	\hat{\mu}_{\lambda,j}$ and $\hat{\mu}_{\lambda,j+k} \neq \hat{\mu}_{\lambda,j+k+1}$.
	We need to show that the KKT condition \eqref{eq:KKT} admits the solution 
	\begin{align}
		& \hat{\mu}_{\lambda,j} = \hat{\mu}_{\lambda,j+1} = \cdots = 
		\hat{\mu}_{\lambda,j+k}, \label{kktsol1} \\
		& s_{\lambda,j},s_{\lambda,j+1},\dots,s_{\lambda,j+k-1} \in [0,1] \label{kktsol2}
	\end{align}
	for $\lambda \geq \lambda_0$.
	Below, assuming \eqref{kktsol1} for $\lambda\geq\tilde{\lambda}$ and \eqref{kktsol2} for $\lambda = \tilde{\lambda}$, we show that the corresponding $s_{\lambda,i}$ satisfy \eqref{kktsol2} for $\lambda > \tilde{\lambda}$.
	
	From the KKT condition \eqref{eq:KKT}, we have 
	$w_i (\hat{\mu}_{\lambda,i} - x_i) + \lambda(s_i-s_{i-1})=0$ for $i=j,\dots,j+k$, and this relation can be rewritten as $w_{i+1} (\hat{\mu}_{\lambda,i+1} - x_{i+1}) + \lambda(s_{i+1}-s_{i})=0$ for $i=j-1,\dots,j+k-1$.
	For $\lambda\geq \tilde{\lambda}$,
	multiplying these two expressions by $w_{i+1}$ and $w_i$, respectively, and considering the subtraction lead to
	\begin{align}
		& \underbrace{\begin{bmatrix}
				w_j + w_{j+1} & -w_j & & &  \\
				-w_{j+2} & w_{j+1} + w_{j+2} & -w_{j+1} & & \\
				& -w_{j+3} & w_{j+2} + w_{j+3} & -w_{j+2} & \\
				& &\ddots &\ddots &\ddots \\
				& & & -w_{j+k} & w_{j+k-1} + w_{j+k} & w_{j+k-1}
		\end{bmatrix}}_{=:A \in \bbR^{k\times k}}
		\underbrace{\begin{bmatrix}
				s_{\lambda,j} \\ s_{\lambda,j+1} \\ \vdots \\  \vdots \\ s_{\lambda,j+k-1} 
		\end{bmatrix}}_{=:s_\lambda \in\bbR^k} \\
		& \quad =  \frac{1}{\lambda}
		\underbrace{\begin{bmatrix}
				w_j w_{j+1} & - w_j w_{j+1} & & & & \\
				& w_{j+1} w_{j+2} & - w_{j+1} w_{j+2} & & &  \\
				& & & \ddots & \ddots & \\
				& & & & w_{j+k-1} w_{j+k} & - w_{j+k-1} w_{j+k} 
		\end{bmatrix}}_{=: D \in \bbR^{k\times (k+1)}}
		\underbrace{\begin{bmatrix}
				X_{j} \\ X_{j+1} \\ \vdots \\ X_{j+k}
		\end{bmatrix}}_{=:y\in\bbR^{k+1}} \\
		& \quad \phantom{=} \quad 
		+ \underbrace{\begin{bmatrix}
				w_{j+1} & & & & \\ 
				& 0 & & & \\
				& & \ddots & & \\
				& & & 0 & \\
				& & & & w_{j+k-1}
		\end{bmatrix}}_{=:E\in\bbR^{k\times k}}
		\underbrace{\begin{bmatrix}
				s_{\lambda,j-1} \\ 0 \\ \vdots \\ 0 \\ s_{\lambda,j+k} 
		\end{bmatrix}}_{=:c_\lambda \in\bbR^k},
	\end{align}
	where the assumption \eqref{kktsol1} is used.
	It is easy to show that $A$ is non-singular when all weights are positive; thus, we have
	\begin{equation*}
		s_\lambda = \frac{1}{\lambda} A^{-1} D y + A^{-1} E c_\lambda .
	\end{equation*}
	Since we have assumed \eqref{kktsol2} for $\lambda = \tilde{\lambda}$,
	all elements of $s_\lambda$ are in $[0,1]$ when $\lambda = \tilde{\lambda}$.
	It remains to show that all elements of $s_\lambda$ remain in $[0,1]$ when $\lambda \geq \tilde{\lambda}$.
	As $\lambda$ increases, the first term of the right-hand-side gets smaller in magnitude. 
	Therefore, if $A^{-1}Ec_\lambda$ is in $[0, 1]$ coordinate-wise, then the right-hand-side will stay in $[0, 1]$ for increasing $\lambda$.
	Below we show that every element of $A^{-1}Ec_\lambda$ is in $[0,1]$.
	
	Note that the first and last elements of $c_\lambda$ is either $0$ or $1$, and all elements of $A^{-1} E$ except for the first and last ($k$-th) columns are zero.
	We will check that every element of the first and last columns of $A^{-1} E$ is positive, and $(A^{-1}E)_{i1} + (A^{-1}E)_{ik} = 1$, which readily indicates that $A^{-1}Ec_\lambda$ is in $[0, 1]$.
	By Cramer's rule, we have
	\begin{equation}
		(A^{-1}E)_{i1} = \frac{\begin{vmatrix} 
				& & & w_{j+1} & & \\
				& & & 0 & & \\
				a_1 & a_2 & \cdots & \vdots & \cdots & a_k \\
				& & & 0 & & \\
				& & & 0 & &
		\end{vmatrix}}{|A|}, \quad
		(A^{-1}E)_{ik} = \frac{\begin{vmatrix} 
				& & & 0 & & \\
				& & & 0 & & \\
				a_1 & a_2 & \cdots & \vdots & \cdots & a_k \\
				& & & 0 & & \\
				& & & w_{j+k} & &
		\end{vmatrix}}{|A|},
	\end{equation}
	where $|\cdot|$ denotes the determinant of a matrix, and $a_i$ denotes the $i$-th column of $A$.
	Here, the numerators and denominator $|A|$ are positive, which can be proved by induction.
	Thus, every element of the first and last columns of $A^{-1} E$ is positive.
	Further, since $(w_{j+1},0,\dots,0,w_{j+k})^\top = \sum_{i=1}^k a_k$,
	it follows that
	\begin{equation}
		(A^{-1}E)_{i1} + (A^{-1}E)_{ik} = \frac{\begin{vmatrix} 
				& & & w_{j+1} & & \\
				& & & 0 & & \\
				a_1 & a_2 & \cdots & \vdots & \cdots & a_k \\
				& & & 0 & & \\
				& & & w_{j+k} & &
		\end{vmatrix}}{|A|}
		= \frac{|A|}{|A|} = 1.
	\end{equation}
\end{proof}


\section{Proof of Proposition 1}
\begin{proof}
	We consider the case of $\beta=\infty$ without loss of generality.
	Since \eqref{bounded} is a convex program, the necessary and sufficient condition for its optimal solution is given by the KKT condition \citep[Section~5.5.3]{bv04}:
	\begin{align*}
		-x_i + \psi'(\theta_i) + \lambda (\rho_i-\rho_{i-1}) + \nu_i &= 0, \\
		\nu_i(\theta_i-\alpha)&=0, \\
		\rho_i &
		\begin{cases}
			= 1 & (\theta_i > \theta_{i+1}) \\
			= 0 & (\theta_i < \theta_{i+1})\\
			\in [0,1] & (\theta_i = \theta_{i+1})
		\end{cases}, \\
		\theta_i & \geq \alpha, \\
		\nu_i & \geq 0
	\end{align*}
	for $i=1,\dots,n$.
	From \eqref{KKT1} and \eqref{KKT2}, it is satisfied by taking $\theta_i = \max((\hat{\theta}_{\lambda})_i,\alpha)$, $\rho_i=\xi_i$ and
	\begin{align*}
		\nu_i =
		\begin{cases}
			0 & (\theta_i > \alpha) \\
			\psi'((\hat{\theta}_{\lambda})_i)-\psi'(\alpha) & (\theta_i = \alpha)
		\end{cases}
	\end{align*}
	for $i=1,\dots,n$.
	Note that $\nu_i \geq 0$ since $\psi$ is convex and thus $\psi'$ is monotone increasing. 
\end{proof}

\section{Simulation result for chi-square}
\label{subsec:chi-square}
We check the performance of the proposed method for the chi-square distribution.
For $i=1,\dots,100$, let $X_i \sim s_i \chi^2(d_i)$ be a sample from the chi-square distribution with $d_i$ degrees of freedom, where $s_1,\dots,s_{100}$ is a piecewise monotone sequence defined by
\begin{align*}
	s_i = \begin{cases} 1 + 9 \cdot \frac{i-1}{49} & (i=1,\dots,50) \\ 1 + 9 \cdot \frac{i-51}{49} & (i=51,\dots,100) \end{cases}.
\end{align*}
We apply the proposed method to estimate $s_1,\dots,s_{100}$ from $X_1,\dots,X_{100}$.

First, we set $d_i=5$ for $i=1,\dots,100$.
Figure~\ref{fig:chisquare} shows $\hat{s}_{\lambda}$ for several knot values of $\lambda$.
Similarly to the original nearly isotonic regression, the estimate is piecewise monotone and the number of joined pieces decreases as $\lambda$ increases.
In this case, $\hat{s}_{\lambda}$ becomes monotone at the final knot $\lambda = 270.04$ and it coincides with the result of the proposed method.
Figure~\ref{fig:chisquare_AIC} plots ${\rm AIC}(\lambda)$ with respect to $\lambda$.
It takes minimum at $\hat{\lambda}=80.68$, which corresponds to the third panel of Figure~\ref{fig:chisquare}.
In this way, the proposed information criterion enables to detect change-points in the parameter sequence in a data-driven manner.

\input{chisquare.tex}

\begin{figure}[htbp]
\centering
\begin{tikzpicture}
\tikzstyle{every node}=[]
\begin{axis}[
width=7cm,
xlabel={$\lambda$},
ylabel={${\rm AIC}(\lambda)$},
]
\addplot[very thick, color = blue, opacity=1,
] table {
0 889.2815157622417
0.19167960017133368 887.4315125341731
0.2900041583501878 885.6033738122727
0.6310232478728657 884.543163488562
0.8150263728565306 883.1915797218596
1.2912892527345252 883.2041541247534
1.7288203842377494 883.3523935122727
2.0056280599653715 882.8222269007733
2.606114892573452 884.1244699791847
2.6402791829544054 882.3157480422294
3.0111605130177534 882.4359313423763
3.0614097169377694 880.6454022164586
3.1376939036729916 878.944797685622
3.6424280049674462 878.9483556266331
4.249351252683464 879.4463755029911
4.725119110587546 879.4836872142064
4.9339176625230365 878.3953625514945
5.1455488272040135 877.3226939939008
5.363725433878596 876.290640430937
5.491657872672317 874.8565010332483
6.187732693132779 878.0148360152783
6.467892091924685 875.1461592061504
7.076936618259155 875.6780091680854
7.335932527822142 874.7751800078951
7.530417801116585 873.6104072852586
7.804810272203562 872.7789511190026
7.9346315177929165 871.3358834401437
8.015367317890327 871.684093326465
8.027035167003934 867.7329964288857
8.146335114934386 866.2016201084477
8.810792181471916 866.8236869613662
8.958078838299025 865.3405682339929
8.977821036526247 863.4063062929564
8.984351053090503 861.4270634223865
9.018382295339336 859.5268970799802
9.488929931755028 858.7974272109168
10.41286240421714 859.3838485701522
11.441347297731358 862.2173932091459
11.448557294301285 858.235303773668
12.377680047871824 858.4535670826947
13.471275118722122 861.1341428575229
13.525646242441809 857.2396929876943
15.167242067364768 858.2566604886606
15.516671383226686 856.8989465337372
15.929441265355063 855.6564460149403
16.015551794893835 853.8018271151489
16.154101584799598 852.0263750718711
16.199276366084284 850.0921546283064
17.207300166060946 849.5534970635358
17.310415564628006 847.7024544736717
18.01523052041775 846.6871718260645
18.701355600567386 845.5978138397483
18.773669362853003 843.6940745536908
19.096442375520194 842.1169802427985
20.648183812211673 842.0022507567272
20.79220923052878 840.167004456553
22.124931602671975 839.6548411723603
25.2768265388451 841.3852137817457
25.73675129416817 839.9615378532781
25.86961813306011 838.1093838063257
26.876163463823044 837.197888988566
27.117174093817603 835.4455712057874
31.069725587040047 836.875276237086
31.562031508599198 835.3139520130952
31.657355099547686 833.3969147975124
32.923269181962816 832.2761320567387
34.13101706277677 830.9108335947828
36.59782962831848 830.2004724968551
36.60620030533941 828.2048849213169
39.55330192044022 827.5403959134733
41.31856094010499 826.3446038339198
51.436094604680264 829.119514807633
53.75107143304206 830.250523557593
57.02474702791596 827.6531541463445
60.74963599119365 827.1135459160673
67.49442406868506 828.0590403293606
70.33140337657312 827.1442247078384
76.38820096426247 827.0891653250009
76.8272597114776 825.2357323855337
80.68680806399138 824.5631594876462
83.79743217564044 825.6970012513924
117.12287396435164 827.309386518126
147.93475965303367 830.9570868694765
148.8055375293766 829.1184287762968
155.27780215294948 828.2193335686726
261.79220925259415 844.266865123351
270.0494438403357 843.3172597170166
};
\end{axis}
\end{tikzpicture} 
\caption{AIC for the chi-square distribution ($d=5$).}
\label{fig:chisquare_AIC}
\end{figure}

Next, we set $d_i=d$ for $i=1,\dots,100$ with $d \in \{ 2,3,5,10 \}$.
Figure~\ref{fig:chisquare_bias} plots ${\rm E}_{\theta}[{\rm AIC}(\lambda)]$ and $2 {\rm E}_{\theta} [L(\theta,\hat{\theta}_{\lambda})]$ with respect to $\lambda$ for each value of $d$, where we used 10000 repetitions.
They take minimum at similar values of $\lambda$.
The absolute bias $|{\rm E}_{\theta}[{\rm AIC}(\lambda)]-2 {\rm E}_{\theta} [D(\theta,\hat{\theta}_{\lambda})]|$ decreases as $d$ increases, which is compatible with the fact that the chi-square distribution becomes closer to the normal distribution for larger $d$.

\input{chisquare_bias.tex}

Finally, we examine the case where the degrees of freedom are not constant:
\begin{align*}
	d_i = \begin{cases} 6 & (i=1,6,\dots,96) \\ 7 & (i=2,7,\dots,97) \\ 8 & (i=3,8,\dots,98) \\ 9 & (i=4,9,\dots,99) \\ 10 & (i=5,10,\dots,100) \end{cases}.
\end{align*}
Figure~\ref{fig:chisquare_bias_weighted} plots ${\rm E}_{\theta}[{\rm AIC}(\lambda)]$ and $2 {\rm E}_{\theta} [D(\theta,\hat{\theta}_{\lambda})]$ with respect to $\lambda$, where we used 10000 repetitions.
The bias of the proposed information criterion is sufficiently small.
Thus, this criterion works well for determining the regularization parameter $\lambda$ even when the degrees of freedom are heterogeneous among samples.

\input{chisquare_bias_weighted.tex}

\end{document}